\documentclass[aps,pra,superscriptaddress,showpacs,twocolumn,nofootinbib]{revtex4-2}

\usepackage{soul}
\setstcolor{red}

\usepackage{multirow}    

\usepackage{caption}
\DeclareCaptionJustification{justified}{\leftskip=0pt \rightskip=0pt \parfillskip=0pt plus 1fil}

\usepackage{graphicx}
\usepackage{dcolumn}
\usepackage{bm}
\usepackage{caption}
\usepackage{amsmath,amssymb,amsthm}
\usepackage{braket}
\usepackage{empheq}
\usepackage{subfig}
\usepackage{tikz}
\usepackage{hyperref}
\usepackage{xcolor}
\usepackage{float}
\DeclareMathOperator{\Tr}{Tr}
\usepackage{booktabs} 
\usepackage{makecell} 

\hypersetup{
    colorlinks=true,       
    linkcolor=blue,          
    citecolor=red,        
    filecolor=red,      
    urlcolor=blue,           
    runcolor=cyan
}

\newtheorem{theorem}{Theorem}
\newtheorem*{theorem*}{Theorem}

\newtheorem{lemma}{Lemma}
\newtheorem{proposition}{Proposition}

\usepackage{adjustbox}

\begin{document}

\title{Robust high-order quantum simulation using finite-width pulses}

\author{Leeseok Kim}
\affiliation{Department of Electrical \& Computer Engineering and Center for Quantum Information and Control, University of New Mexico, Albuquerque, NM 87131, USA}

\author{Milad Marvian}
\affiliation{Department of Electrical \& Computer Engineering and Center for Quantum Information and Control, University of New Mexico, Albuquerque, NM 87131, USA}

\begin{abstract}

We present a general framework for promoting first-order pulse sequences in quantum simulation to higher-order sequences that maintain robustness in the presence of finite pulse-width effects. Our approach maps a given first-order pulse sequence to a first-order Trotter formula, applies higher-order Trotter-formula constructions, and then compiles the resulting evolution back into physically implementable finite-width pulses via dynamically corrected gates.
The resulting sequences achieve arbitrarily high-order error scaling with respect to the control cycle time of the underlying first-order sequence while maintaining robustness to finite pulse-width effects. The framework also enables the use of multi-product formulas for more efficient constructions. We apply the framework to several physically motivated quantum-simulation tasks and numerically verify the predicted error scalings.

\end{abstract}

\maketitle

\section{Introduction}
Quantum simulation aims to realize the unitary evolution generated by a target Hamiltonian on quantum processors \cite{feynman1982simulating,georgescu2014quantum}. This task is an important application of quantum computing, as many such simulations are expected to be computationally intractable on classical computers \cite{lloyd1996universal}. Therefore, substantial effort has been devoted to developing efficient quantum simulation protocols in both the digital and analog paradigms. In the digital approach, the target dynamics are compiled into a sequence of one- and two-qubit gates from a universal gate set. In contrast, the analog approach uses a controllable physical system whose native interactions, together with designed controls, engineer an effective Hamiltonian close to the target, thereby realizing the desired evolution through continuous-time dynamics.

Here we focus on the analog setting, where designing suitable controls is naturally cast as a quantum-control problem. A convenient and widely used viewpoint is pulse-based control, in which fast control pulses are interleaved with evolution under the native interactions. Over each control cycle, the net evolution defines an effective Hamiltonian that approximates the target with high accuracy \cite{haeberlen1968coherent}. Such a framework appears across different communities under various names. One prominent example is Hamiltonian engineering, which has a rich history in NMR: it originated in designing pulse sequences to decouple or reshape interactions \cite{waugh1968approach,haeberlen1968coherent,mansfield1971symmetrized,rhim1973analysis,maricq1982application,raleigh1988rotational}, and was later developed to enable efficient universal quantum simulation using the native Hamiltonians available in NMR systems \cite{leung2000efficient,stollsteimer2001suppression,leung2002simulation,dodd2002universal,vandersypen2005NMR}. More recently, Hamiltonian engineering has extended beyond NMR to design pulse sequences across diverse platforms such as Rydberg-atom arrays \cite{bluvstein2021controlling,geier2021floquet,scholl2022microwave,zhao2023floquet-tailored}, superconducting circuits \cite{nguyen2024programmable}, trapped-ion chains \cite{morong2023engineering} and solid-state spin systems \cite{choi2020robust,zhou2023robust_sensing,zhou2024robust}. In parallel, efficient theoretical methods for systematic Hamiltonian engineering have been proposed \cite{hayes2014programmable,bookatz2014hamiltonian,bookatz2016improved,haas2019engineering,babler2023synthesis,Votto2024universal,babler2024time,babler2025general,chen2026engineering}.
Another example is digital-analog quantum simulation \cite{arrazola2016digitalanalog,lamata2017digitalanalog,lamata2018digitalanalog,parrarodriguez2020digitalanalog,gonzalezraya2021digitalanalog,yu2022superconducting,andersen2025thermalization}, which alternates analog evolution under native multi-qubit interactions with digital layers, typically fast single-qubit rotations, to implement the target dynamics. This approach has been experimentally demonstrated in a variety of recent quantum simulation protocols, including implementations on trapped-ion \cite{arrazola2016digitalanalog,kumar2025digitalanalogb,katz2025hybrid} and superconducting platforms \cite{lamata2018digitalanalog,tao2021experimental,kumar2025digitalanaloga,andersen2025thermalization}.

This fast pulse-based simulation framework is most often treated only to leading order in the control-cycle time, since constructing pulse sequences with higher-order accuracy is typically challenging. However, higher-order designs are crucial for more accurate quantum simulation. Moreover, many analyses assume ideal instantaneous (bang–bang) control pulses. In practice, control pulses have finite duration, and the native Hamiltonian remains active during their application, leading to systematic distortions of the intended control operation. These errors accumulate over long sequences and can significantly degrade simulation performance. Developing pulse-sequence frameworks that provide higher-order accuracy while remaining robust to finite-width pulse effects is therefore an important yet still underexplored problem.

In this work, we present a systematic method that promotes any first-order instantaneous-pulse sequence to pulse sequences achieving arbitrarily high-order simulation while remaining robust to finite pulse-width errors. The construction proceeds in four steps: we reinterpret a given sequence as a first-order Trotter formula for suitable effective Hamiltonians, upgrade it via higher-order Trotter formulas, map the resulting exponentials back to native Hamiltonians and control pulses, and implement each pulse using dynamically corrected gates (DCGs) \cite{khodjasteh2009dynamical,khodjasteh2009dynamically,khodjasteh2010arbitrarily} by treating the always-on native Hamiltonian as noise. In the ideal-pulse limit, the resulting sequence achieves arbitrarily high-order accuracy in the control-cycle time, as guaranteed by the underlying high-order Trotter construction. When implemented with finite-width pulses, the constructed sequence is typically longer than the original, since in our framework the pulse count is proportional to the products of exponentials used in high-order Trotter formulas, which grows exponentially with the desired error order in the control-cycle time. While this growth could, in principle, amplify control imperfections, DCGs strongly suppress this error to keep the overall simulation error manageable. As illustrative examples, we convert several known first-order sequences into higher-order ones (e.g., fourth order), achieving improved simulation accuracy with realistic pulse widths.

Our framework applies not only to first-order sequences built from instantaneous pulses, but also to first-order sequences that are already robust to finite pulse widths \cite{bookatz2014hamiltonian,choi2020robust,peng2022deep,tyler2023higher,zhou2024robust,Votto2024universal,babler2025general}. These schemes explicitly model finite-width pulses and design the sequence so that it both reproduces the target dynamics to first-order and averages away pulse-width error contributions. Our framework still produces higher-order sequences in this regime, and the baseline sequence’s built-in robustness allows us to realize the required operations simply by composing the same pulses already used in the original design, without introducing any new DCG constructions.

We note that applying our construction at fourth order and beyond requires negative-time evolution under the native Hamiltonian, due to the negative coefficients that arise in higher-order Trotter formulas. Although advanced dynamical-decoupling techniques can realize such time reversal, doing so typically incurs a pulse overhead that grows exponentially with system size. Motivated by this challenge, we show that in applications where the goal is to estimate expectation values of observables, rather than to prepare the full quantum state, substantially more efficient techniques can be developed. In this setting, we construct pulse sequences based on multiproduct formulas \cite{childs2012hamiltonian,carrera2023well} that avoid negative-time evolution altogether, as they are built exclusively from second-order Trotter formulas.

\section{Summary of main results}\label{sec2}

We consider an $n$-qubit system governed by
\begin{align}\label{eq:model}
H(t)=H_0+H_c(t),
\end{align}
where $H_0$ is the time-independent native Hamiltonian and $H_c(t)$ is a control Hamiltonian generated by classical control functions. Our goal is to implement the target evolution $e^{-iH_{\mathrm{targ}}T}$ using a single control cycle of duration $T_c$,
\begin{align}\label{eq:task}
U(T_c):=\mathcal{T}\exp\left(-i\int_{0}^{T_c} H(t) dt\right)\approx e^{-iH_{\mathrm{targ}}T},
\end{align}
where $\mathcal{T}$ denotes time ordering. The cycle time $T_c$ may differ from the simulation time $T$. It suffices to design a short cycle since repeating it $r$ times simulates a longer evolution and the total error scales at most linearly in $r$.

To specify a pulse sequence, we introduce the control propagator
\begin{align}\label{eq:Uc_def}
U_c(t)=\mathcal{T}\exp\left(-i\int_{0}^{t} H_c(t')dt'\right).
\end{align}
In the interaction frame defined by $U_c(t)$, the evolution over one cycle can be written as
\begin{align}\label{eq:interaction_frame}
U(T_c)=U_c(T_c) \mathcal{T}\exp\left(-i\int_{0}^{T_c} U_c^\dagger(t)H_0U_c(t) dt\right).
\end{align}
We restrict to cyclic controls satisfying $U_c(T_c)=I$. If the average condition
\begin{align}\label{eq:avg_condition}
\int_{0}^{T_c} U_c^\dagger(t)H_0U_c(t)dt = H_{\mathrm{targ}}T
\end{align}
holds, then under a convergence condition for Magnus expansion such as $\Lambda T_c<\pi$ \cite{moan2008convergence},
\begin{align}\label{eq:first_order_base}
U(T_c)=e^{-iH_{\mathrm{targ}}T}+\mathcal{O}\big((\Lambda T_c)^2\big),
\end{align}
where $\Lambda:=\|H_0\|$ and $\| \cdot \|$ denotes the operator norm. We call a pulse sequence \emph{first-order} if it satisfies Eq.~\eqref{eq:first_order_base}.

Our first result upgrades any first-order instantaneous pulse sequence to arbitrarily high-order in $\Lambda T_c$ in the ideal pulse limit while remaining robust to finite pulse duration $t_p$. Let $l$ be the number of pulses in the given first-order sequence. For any integers $p\ge 1$ and $q\ge 1$, we construct a sequence $\widetilde{\mathcal S}_{2p}(T)$ such that
\begin{align}\label{eq:result1}
\begin{aligned}
\bigl\|\widetilde{\mathcal S}_{2p}(T)-e^{-iH_{\mathrm{targ}}T}\bigr\| & = \mathcal O\bigl((\Lambda T_c)^{2p+1}\bigr) \\ 
&\quad + \mathcal O\bigl(5^{p-1} l \varepsilon_{\mathrm{DCG}}^{(q)}\bigr),
\end{aligned}
\end{align}
where $\varepsilon_{\mathrm{DCG}}^{(q)}=\mathcal O\bigl((\kappa^{q}\Lambda t_p)^{q+1}\bigr)$ denotes the error incurred when implementing one ideal pulse by a $q$th-order concatenated dynamically corrected gate construction \cite{khodjasteh2009dynamical,khodjasteh2010arbitrarily}. The constant $\kappa>1$ is determined by the available control resources and by how the error induced by $H_0$ is modulated by the applied controls during a finite-width pulse. In particular, the second term in Eq.~\eqref{eq:result1} is $(q+1)$th-order in $t_p$, demonstrating robustness to finite pulse widths. Its prefactor grows with the $\mathcal O(5^{p-1}l)$ pulse insertions in the higher-order sequence, so increasing $p$ suppresses the Trotter error at the cost of a longer sequence and larger accumulated DCG error.


Our second result takes as input any first-order pulse sequence specifically designed for implementation using finite-width pulses. Let $l$ be the number of pulses in the given sequence and let each pulse have duration $t_p$. For any integer $p \ge 1$, we construct a sequence $\widetilde{\mathcal S}_{2p}(T)$ using only the same set of control pulses such that
\begin{align}\label{eq:result2}
\begin{aligned}
\bigl\|\widetilde{\mathcal S}_{2p}(T)-e^{-iH_{\mathrm{targ}}T}\bigr\| & = \mathcal O \bigl((\Lambda (T_c+(c_p-1) l t_p))^{2p+1}\bigr) \\
& \quad + \mathcal O\bigl(5^{p-1} p^2 l(l\Lambda t_p)^2\bigr),
\end{aligned}
\end{align}
where $c_p=2\prod_{j=2}^{p}(4-4^{1/(2j-1)})$. When $c_p l t_p \ll T_c$, the leading high-order scaling is essentially unchanged. The second term is second order in $t_p$, again reflecting robustness to finite pulse widths. Unlike the first construction, this procedure uses only pulses already present in the original first-order design.

Both constructions assume access to uniformly stretched and time-reversed implementations of each available pulse shape, as formalized later. In Construction~1, a $q$th-order DCG may require stretching by a factor of up to $2^{\sum_{r=1}^{q}1/r}$. In Construction~2, no new pulse shapes are introduced: all required operations are built from stretched and time-reversed versions of the original pulses, and the worst-case stretch factor in the order-$2p$ sequence is $4^{\sum_{r=1}^{p-1}1/(2r+1)}$.

For $p\ge 2$, both constructions require negative-time evolution under the native Hamiltonian, i.e., segments of the form $e^{+iH_0\tau}$, because higher-order Trotter formulas involve negative coefficients in their recursive compositions. Such operations can be implemented robustly using concatenated dynamical decoupling, but at the cost of exponentially many pulses in $n$. When only expectation values are needed, however, negative-time evolution can be avoided by adapting multi-product formulas \cite{childs2012hamiltonian,carrera2023well} to the pulse-sequence setting. The idea is to classically combine measurement outcomes from several robust second-order pulse sequences, each using only positive-time evolution, to recover error scalings similar to those in Eqs.~\eqref{eq:result1} and~\eqref{eq:result2}.

To benchmark our method, we apply it to several physically motivated quantum-simulation tasks previously addressed with first-order pulse sequences \cite{bookatz2014hamiltonian,parrarodriguez2020digitalanalog,gonzalezraya2021digitalanalog}, systematically upgrade those sequences to higher order, and numerically verify the predicted error scalings.

\section{Background and notation}\label{sec:background-and-notation}
This section reviews quantum simulation by fast pulse sequences based on the Magnus expansion, for both ideal and finite-width pulses, and introduces notation used throughout the paper. We also define time-stretching and pulse reversal, and review the Trotter formulas used later in our constructions. Throughout, we use the convention $\prod_{i=a}^b A_i := A_b A_{b-1} \cdots A_a$.

\subsection{Quantum simulation using fast pulse sequences}\label{subsec:qsim-using-fast-pulse-sequence}

\subsubsection{Ideal pulse sequences}
Given the native and control Hamiltonians in Eq.~\eqref{eq:model}, we consider quantum simulation by a sequence of $l$ instantaneous control pulses separated by free evolution under $H_0$. An ideal pulse sequence is generated by choosing
\begin{align}
    H_c(t) = \sum_{k=1}^l \theta_k \delta(t-t_k)H_{P_k}, \quad t_k = \sum_{j=1}^k \tau_j,
\end{align} 
where $\delta(t)$ is the Dirac delta function and $T_c=t_l=\sum_{k=1}^l \tau_k$ is the total control time. The pulse applied at $t=t_k$ generates the unitary $P_k=e^{-i\theta_k H_{P_k}}$. The resulting propagator is (cf.~Fig.~\ref{fig:Fig1}(a))
\begin{align}\label{Ideal_DD}
\begin{split}
    U(T_c)
    &= \mathcal{T}\exp\left(-i\int_0^{T_c} H(t)dt\right) \\
    &= P_l e^{-iH_0\tau_l} P_{l-1}\cdots P_2 e^{-iH_0\tau_2} P_1 e^{-iH_0\tau_1}.
\end{split}
\end{align}
Given a target Hamiltonian $H_{\mathrm{targ}}$ and simulation time $T$, the goal is to choose $\{P_k\}$ and $\{\tau_k\}$ such that
\begin{align}
    U(T_c)\approx e^{-iH_{\mathrm{targ}}T}.
\end{align}
To analyze the sequence, it is convenient to introduce the control propagator
\begin{align}
    U_c(t)=\mathcal{T}\exp\left(-i\int_0^t H_c(t')dt'\right),
\end{align}
and define
\begin{align}
    g_k:=P_kP_{k-1}\cdots P_1,\qquad g_0:=I.
\end{align}
For ideal pulses, $U_c(t)$ is piecewise constant, with $U_c(t)=g_k$ for $t\in(t_k,t_{k+1})$. Accordingly, the toggling-frame Hamiltonian
\begin{align}
    H_I(t):=U_c^\dagger(t)H_0U_c(t)
\end{align}
is also piecewise constant, namely $H_I(t)=g_k^\dagger H_0 g_k$ on $t\in(t_k,t_{k+1})$. If the control cycle is closed, i.e., $U_c(T_c)=g_l=I$, then
\begin{align}\label{Ideal_DD_toggling}
    U(T_c)
    &=\mathcal{T}\exp\left(-i\int_0^{T_c} H_I(t) dt\right)\\
    &=\prod_{k=0}^{l-1} e^{-i g_k^\dagger H_0 g_k \tau_{k+1}}.
\end{align}
The effective evolution can be characterized using the Magnus expansion \cite{magnus1954on},
\begin{align}
    U(T_c)=e^{-i\Omega},
    \qquad
    \Omega=\sum_{m=1}^\infty \Omega^{(m)},
\end{align}
which converges for $\Lambda T_c<\pi$ \cite{moan2008convergence}. The first-order term is
\begin{align}
    \Omega^{(1)}=\sum_{k=0}^{l-1} g_k^\dagger H_0 g_k \tau_{k+1},
    \label{ideal_Omega1}
\end{align}
while, more generally, $\|\Omega^{(m)}\|=\mathcal{O}\big((\Lambda T_c)^m\big)$. The simulation error can be upper bounded as
\begin{align}
    \|U(T_c)-e^{-iH_{\mathrm{targ}}T}\| \le \|\Omega-H_{\mathrm{targ}}T\|.
\end{align}
Therefore, if
\begin{align}
    \Omega^{(1)} &= H_{\mathrm{targ}}T,\\
    \Omega^{(2)} &= \Omega^{(3)}=\cdots=\Omega^{(m)}=0,
\end{align}
then the simulation error scales as $\mathcal{O}\big((\Lambda T_c)^{m+1}\big)$. Furthermore, if the toggling-frame Hamiltonian is time-symmetric,
\begin{align}
H_I(T_c-t)=H_I(t), \qquad t\in[0,T_c],
\end{align}
then all even-order Magnus terms vanish \cite{iserles2001time}. Hence, a time-symmetric sequence with $\Omega^{(1)}=H_{\mathrm{targ}}T$ achieves $\mathcal{O}\big((\Lambda T_c)^3\big)$ simulation error.

\begin{figure}[t!]
    \centering
    \includegraphics[width=8.6cm]{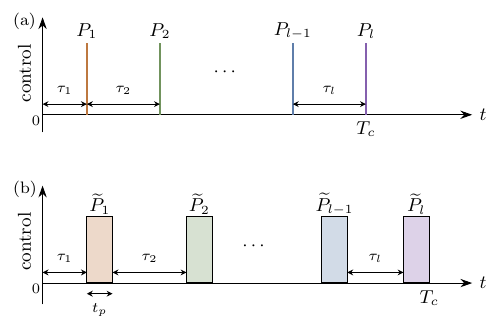}
    \caption{ Pulse sequence with $l$ control pulses. (a) Ideal pulses, where the $k$th pulse is instantaneous and implements $P_k$. (b) Finite-width pulses of duration $t_p$, where the $k$th pulse implements $\widetilde{P}_k$. $T_c$ denotes the control cycle time.}
    \label{fig:Fig1}
\end{figure}

\subsubsection{Finite-width pulse sequences}
In practice, control pulses have finite amplitude and bandwidth, and hence nonzero duration $t_p>0$. We model the $k$th pulse by
\begin{align}\label{H_c(t)}
    H_c(t)=\sum_{k=1}^l f_k(t-s_k)H_{P_k},
    \quad
    s_k:=\sum_{j=1}^k \tau_j + (k-1)t_p,
\end{align}
where each control function $f_k(t)$ is nonzero only on $[0,t_p]$ and satisfies $\int_0^{t_p} f_k(t)dt=\theta_k$. The total control time is now $T_c=\sum_{k=1}^l \tau_k + lt_p$. The unitary generated during the $k$th pulse is
\begin{align}\label{imperfect_pulse_1}
    \widetilde{P}_k =\mathcal{T}\exp\left(-i\int_0^{t_p} \big(H_0+f_k(t)H_{P_k}\big)dt\right),
\end{align}
so that the full propagator becomes
\begin{align}\label{Noisy_DD}
\begin{split}
    U(T_c) &= \widetilde{P}_l e^{-iH_0\tau_l} \widetilde{P}_{l-1}\cdots \widetilde{P}_2 e^{-iH_0\tau_2}\widetilde{P}_1 e^{-iH_0\tau_1}.
\end{split}
\end{align}

It is convenient to separate each finite-width pulse into the ideal pulse $P_k$ and an error unitary generated by $H_0$ during the pulse. Defining
\begin{align}
    U_{P_k}(t)=\mathcal{T}\exp\left(-i\int_0^t f_k(t')H_{P_k}dt'\right),
\end{align}
so that $U_{P_k}(t_p)=P_k$, we write
\begin{align}\label{P_tilde_def}
\begin{split}
    \widetilde{P}_k &= P_k\mathcal{T}\exp\left(-i\int_0^{t_p} U_{P_k}^\dagger(t)H_0U_{P_k}(t)dt\right) \\
    &:= P_k e^{-i\Phi_k},
\end{split}
\end{align}
where $\Phi_k=\sum_{m=1}^\infty \Phi_k^{(m)}$ is the Magnus expansion that converges for $\Lambda t_p<\pi$, and
\begin{align}\label{magnus_1}
    \Phi_k^{(1)}=\int_0^{t_p} U_{P_k}^\dagger(t)H_0U_{P_k}(t)dt.
\end{align}
Thus each finite-width pulse differs from its ideal counterpart by $\mathcal{O}(\Lambda t_p)$, so implementing an ideal pulse sequence with finite-width controls generally introduces an additional $\mathcal{O}(l\Lambda t_p)$ error.

As in the ideal case, using the Magnus expansion we write
\begin{align}
    U(T_c)=e^{-i\Omega},
    \qquad
    \Omega=\sum_{m=1}^\infty \Omega^{(m)},
\end{align}
with $\|\Omega^{(m)}\|=\mathcal{O}\big((\Lambda T_c)^m\big)$. The first-order term is
\begin{align}
\Omega^{(1)}
&= \int_0^{T_c} U_c^\dagger(t)H_0U_c(t)dt \\
&= \sum_{k=0}^{l-1} \left(g_k^\dagger H_0 g_k\tau_{k+1} + g_k^\dagger \Phi_{k+1}^{(1)} g_k\right),
\label{finite_width_Omega1}
\end{align}
where the first term is Eq.~\eqref{ideal_Omega1}, and the second term is the leading finite-width correction. Accordingly, 
\begin{align}
    \|U(T_c)-e^{-iH_{\mathrm{targ}}T}\| \le \|\Omega-H_{\mathrm{targ}}T\|.
\end{align}
A convenient sufficient condition for $\Omega^{(1)}=H_{\mathrm{targ}}T$ is to impose
\begin{align}
    \sum_{k=0}^{l-1} g_k^\dagger H_0 g_k\tau_{k+1} &= H_{\mathrm{targ}}T,
    \label{free_evolution_correction}\\
    \sum_{k=0}^{l-1} g_k^\dagger \Phi_{k+1}^{(1)} g_k &= 0,
    \label{pulse_width_correction}
\end{align}
Without Eq.~\eqref{pulse_width_correction}, a naive finite-width implementation incurs a leading $\mathcal{O}(l\Lambda t_p)$ error. Imposing Eq.~\eqref{pulse_width_correction} cancels this term and restores the $\mathcal{O}\big((\Lambda T_c)^2\big)$ scaling of the ideal first-order design. Such sequences have been developed in, e.g., \cite{bookatz2014hamiltonian,choi2020robust,peng2022deep,Votto2024universal,babler2025general}. Moreover, as in the ideal case, a time-symmetric sequence satisfying $\Omega^{(1)}=H_{\mathrm{targ}}T$ achieves $\mathcal{O}\big((\Lambda T_c)^3\big)$ simulation error.

\subsection{Control-function variations}\label{subsec:control-function-variations}

In addition to a base control function $f_k(t)$ on $[0,t_p]$, our constructions use two elementary operations: time-reversal with sign flip, and time-stretching. For a control function $f_k(t)$, let $\widetilde{P}_k$ denote the corresponding finite-width pulse defined in Eq.~\eqref{P_tilde_def}. We then introduce the reversed pulse
\begin{align}
    \widetilde{P}_k^{\mathrm{rev}} = \mathcal{T}\exp\left(-i\int_0^{t_p}\big(H_0-f_k(t_p-t)H_{P_k}\big) dt \right).
\end{align}
In the ideal limit $H_0=0$, the control functions $f_k(t)$ and $-f_k(t_p-t)$ implement $P_k$ and $P_k^\dagger$, respectively.

We also consider stretched versions of these functions. For $c\ge 1$, we define
\begin{align}
    \widetilde{P}_k(ct_p) &= \mathcal{T}\exp\left(-i\int_0^{ct_p}\Big(H_0+\frac{1}{c}f_k\left(\frac{t}{c}\right)H_{P_k}\Big)dt \right),\label{eq:stretched_pulse}\\
    \widetilde{P}_k^{\mathrm{rev}}(ct_p) &= \mathcal{T}\exp\left(-i\int_0^{ct_p}\Big(H_0-\frac{1}{c}f_k\left(t_p-\frac{t}{c}\right)H_{P_k}\Big)dt \right). \label{eq:stretched_rev_pulse}
\end{align}
The factor $1/c$ preserves the pulse area, so stretching changes the pulse duration but not the ideal pulse action. Fig.~\ref{fig:pulse_variation} illustrates these four variants for a smooth asymmetric control function. 

\begin{figure}[t!]
    \centering
    \includegraphics[width=8.6cm]{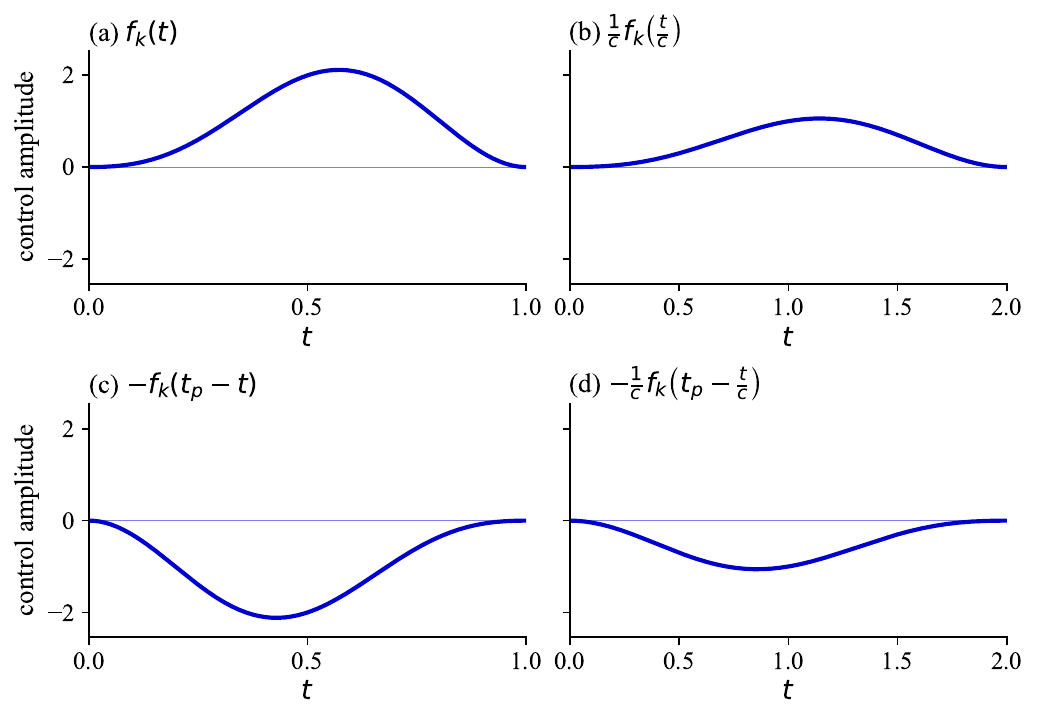}
    \caption{Control-function variations for a smooth asymmetric pulse shape: (a) $f_k(t)$, (b) $f_k(t/c)/c$ with $c=2$, (c) $-f_k(t_p-t)$, and (d) $-f_k(t_p-t/c)/c$.}
    \label{fig:pulse_variation}
\end{figure}

\subsection{Trotter formulas}

A key ingredient in our constructions is the Lie--Trotter--Suzuki product formulas \cite{trotter1959on,suzuki1991general}, which we refer to simply as Trotter formulas.
Let
\begin{align}
    A=\sum_{k=1}^l A_k,
\end{align}
where each $A_k$ is Hermitian. Trotter formulas approximate the target evolution $e^{-iAt}$ by products of the simpler evolutions generated by the individual $A_k$. The first- and second-order formulas are
\begin{align}
    \mathcal{S}_1(t) &= \prod_{k=1}^l e^{-iA_k t},
    \label{Trotter_1st}\\
    \mathcal{S}_2(t) &= \prod_{k=1}^l e^{-iA_k t/2}\prod_{k=l}^1 e^{-iA_k t/2},
    \label{Trotter_2nd}
\end{align}
and the higher even-order formulas are defined recursively by \cite{suzuki1991general}
\begin{align}
    \mathcal{S}_{2p}(t) = \big[\mathcal{S}_{2p-2}(u_p t)\big]^2 \mathcal{S}_{2p-2}\big((1-4u_p)t\big) \big[\mathcal{S}_{2p-2}(u_p t)\big]^2,
\label{Trotter_2kth}
\end{align}
where $u_p=(4-4^{1/(2p-1)})^{-1}$ for $p\ge 2$.

For small $t$, these formulas satisfy \cite{childs2021theory}
\begin{align}
    \|\mathcal{S}_{2p}(t)-e^{-iAt}\| = \mathcal{O}\big(\alpha_{\mathrm{comm}} t^{2p+1}\big),
    \label{Trotter_Error}
\end{align}
where $\alpha_{\mathrm{comm}}$ depends on nested commutators of the terms $\{A_k\}$ up to order $2p+1$,
\begin{align}
    \alpha_{\mathrm{comm}} := \sum_{k_1,\dots,k_{2p+1}=1}^l \left\|[A_{k_{2p+1}},[\cdots,[A_{k_2},A_{k_1}]\cdots]]\right\|.
\end{align}

\section{Methods}\label{sec:methods}

In this section, we present a framework for constructing pulse sequences for robust high-order simulation, and then instantiate it in two constructions. We also discuss the implementation of the negative-time evolutions required by higher-order formulas and the use of multi-product formulas to circumvent them when the goal is to estimate expectation values of observables.

\subsection{General framework}

We begin with a first-order pulse sequence that simulates a target Hamiltonian $H_{\mathrm{targ}}$ for time $T$ using cycle time $T_c$. The sequence may employ either ideal or finite-width pulses. Our goal is to lift such a construction to higher order while retaining robustness to finite pulse duration.

The key observation is that any first-order sequence satisfying $\Omega^{(1)}=H_{\mathrm{targ}}T$ in the ideal-pulse setting as given in Eq.~\eqref{ideal_Omega1}, or Eqs.~\eqref{free_evolution_correction} and \eqref{pulse_width_correction} in the finite-width setting, defines a first-order Trotter formula for suitable operators $\{A_k\}$ satisfying $\sum_k A_k=H_{\mathrm{targ}}$. These operators are determined by the pulse operations and the native Hamiltonian $H_0$. Replacing the first-order product by a higher-order Trotter formula then yields an approximation to $e^{-iH_{\mathrm{targ}}T}$ with improved scaling in $T_c$. Finally, compiling the resulting product back into pulse operations gives a finite-width implementation compatible with the minimum pulse duration $t_p$ and robust to pulse-width errors.

\subsection{Construction 1: Ideal pulse-sequence input}\label{subsec:construction1}

Consider a first-order ideal pulse sequence as in Eq.~\eqref{Ideal_DD},
\begin{align}\label{algo1_starting_point}
U(T_c) = \prod_{k=1}^{l} P_k e^{-iH_0\tau_k},
\end{align}
satisfying $\| U(T_c) - e^{-iH_{\mathrm{targ}}T} \| = \mathcal{O}((\Lambda T_c)^2)$. We upgrade this sequence to higher-order in four steps.

\subsubsection{First-order Trotter mapping}
Using the toggling-frame representation in Eq.~\eqref{Ideal_DD_toggling},
\begin{align}\label{algo1_U_T_c}
U(T_c) = \prod_{k=0}^{l-1} e^{-ig_k^\dagger H_0 g_k\tau_{k+1}},
\end{align}
and the first-order condition $\Omega^{(1)} = H_{\mathrm{targ}}T$ from Eq.~\eqref{ideal_Omega1}, we identify $U(T_c)$ as the first-order Trotter formula $\mathcal{S}_1(T)$ defined in Eq.~\eqref{Trotter_1st} with
\begin{align}\label{exponent_algo1}
A_k = \frac{g_{k-1}^\dagger H_0  g_{k-1} \tau_k}{T}, \quad k=1,\dotsc,l.
\end{align}
By construction, $A = \sum_{k=1}^l A_k = H_{\mathrm{targ}}$.

\subsubsection{Applying higher-order Trotter formulas}
We now approximate $e^{-iH_{\mathrm{targ}}T}$ using higher-order Trotter formulas for the operators $\{A_k\}$ defined in Eq.~\eqref{exponent_algo1}. Applying the second-order formula $\mathcal{S}_2(t)$ from Eq.~\eqref{Trotter_2nd} with $t=\alpha T$ gives
\begin{align}\label{S2_ideal}
\mathcal{S}_2(\alpha T)
&= \prod_{k=0}^{l-1} e^{-ig_k^\dagger H_0 g_k \frac{\alpha \tau_{k+1}}{2}}
   \prod_{k=l-1}^{0} e^{-ig_k^\dagger H_0 g_k \frac{\alpha \tau_{k+1}}{2}}.
\end{align}
For order $2p$ with $p \geq 2$, the recursion in Eq.~\eqref{Trotter_2kth} gives
\begin{align}
    \mathcal{S}_{2p}(T) = [\mathcal{S}_{2p-2}(u_pT)]^2  \mathcal{S}_{2p-2}((1-4u_p)T)  [\mathcal{S}_{2p-2}(u_pT)]^2,
\end{align}
which expands into $5^{p-1}$ copies of $\mathcal{S}_2(\alpha_j T)$, where each coefficient $\alpha_j$ is a product of $p-1$ factors drawn from $\{u_k, 1-4u_k\}_{k=2}^{p}$.

\subsubsection{Compiling back to pulse sequences}
Each copy of $\mathcal{S}_2(\alpha_j T)$ can be rewritten in terms of pulse operations by inverting the toggling-frame mapping. From Eq.~\eqref{S2_ideal},
\begin{align}\label{S2_pulse}
    \mathcal{S}_2(\alpha T) = \prod_{k=1}^{l} P_k e^{-iH_0\frac{\alpha\tau_k}{2}} \prod_{k=l}^{1} e^{-iH_0\frac{\alpha\tau_k}{2}} P_k^\dagger.
\end{align}
Thus each block is a palindromic sequence of $2l$ pulses,
$\{P_1,\dots,P_l,P_l^\dagger,\dots,P_1^\dagger\}$. For $\alpha=1$, this reduces to the standard time-symmetrization of the first-order construction. Consequently, $\mathcal{S}_{2p}(T)$ consists of $5^{p-1}$ such blocks, for a total of $2l\cdot 5^{p-1}$ pulses drawn from $\{P_k,P_k^\dagger\}$.

For $p\geq 2$, some coefficients $\alpha_j$ are negative, since $1-4u_k<0$. The corresponding free-evolution segments $e^{-iH_0\alpha_j\tau_k/2}$ therefore require negative evolution under $H_0$. Their implementation is discussed in Sec.~\ref{sec:negative_time}.

\subsubsection{Robust implementation via dynamically corrected gates}

The sequence $\mathcal{S}_{2p}(T)$ contains $2l\cdot 5^{p-1}$ ideal operations drawn from $\{P_k,P_k^\dagger\}$. A naive finite-width implementation, obtained by replacing each ideal pulse by its finite-width counterpart, introduces an $\mathcal{O}(\Lambda t_p)$ error per pulse and hence a total error of order $\mathcal{O}(5^{p-1}l\Lambda t_p)$.

To suppress this contribution, we replace each ideal gate by a dynamically corrected gate constructed from the stretched and reversed control functions introduced in Sec.~\ref{sec:background-and-notation}, following the DCG construction of Ref.~\cite{khodjasteh2009dynamical,khodjasteh2009dynamically,khodjasteh2010arbitrarily}. For each $W\in\{P_k,P_k^\dagger\}$, let $\widetilde{W}^{[q]}$ denote a $q$th-order DCG implementation satisfying
\begin{align}\label{DCG_error}
    \|\widetilde{W}^{[q]}-W\| = \mathcal{O}\big((\kappa^q\Lambda t_p)^{q+1}\big),
\end{align}
where $q\ge 1$ is the robustness order and $\kappa>1$ is a construction-dependent constant. The explicit DCG constructions used here are reviewed in Appendix~\ref{app:dcg}.

\subsubsection{Error bound}
Let $\widetilde{\mathcal{S}}_{2p}(T)$ denote the sequence obtained from $\mathcal{S}_{2p}(T)$ by replacing each ideal gate $W\in\{P_k,P_k^\dagger\}$ by its DCG implementation $\widetilde{W}^{[q]}$. Then the total error consists of the Trotter error of $\mathcal{S}_{2p}(T)$ and the DCG implementation error.

\begin{theorem}\label{thm_algo1}
The sequence $\widetilde{\mathcal{S}}_{2p}(T)$ obeys the bound
\begin{align}\label{alg1_bigO_clean}
\begin{aligned}
\big\|\widetilde{\mathcal{S}}_{2p}(T)-e^{-iH_{\mathrm{targ}}T}\big\| & = \mathcal{O} \big((\Lambda T_c)^{2p+1}\big) \\
& \quad + \mathcal{O} \big(5^{p-1}l (\kappa^q\Lambda t_p)^{q+1}\big).
\end{aligned}
\end{align}
\end{theorem}
\begin{proof}
By the triangle inequality,
\begin{align}
\begin{aligned}
\big\|\widetilde{\mathcal{S}}_{2p}(T)-e^{-iH_{\mathrm{targ}}T}\big\| & \leq \underbrace{\big\|\mathcal{S}_{2p}(T)-e^{-iH_{\mathrm{targ}}T}\big\|}_{\text{Trotter error}} \\
& \quad + \underbrace{\big\|\widetilde{\mathcal{S}}_{2p}(T)-\mathcal{S}_{2p}(T)\big\|}_{\text{DCG error}}.
\end{aligned}
\end{align}
For the first term, the decomposition~\eqref{exponent_algo1} gives $\|A_k\|=\Lambda\tau_k/T$, so $\alpha_{\mathrm{comm}} \le 2^{2p}\big(\sum_k\|A_k\|\big)^{2p+1} =\mathcal{O}\big((\Lambda T_c/T)^{2p+1}\big)$. Eq.~\eqref{Trotter_Error} therefore yields $\|\mathcal{S}_{2p}(T)-e^{-iH_{\mathrm{targ}}T}\| = \mathcal{O}\big((\Lambda T_c)^{2p+1}\big)$. For the second term, $\mathcal{S}_{2p}(T)$ contains $2l\cdot 5^{p-1}$ ideal pulse operations, each replaced by a DCG with error $\mathcal{O}\big((\kappa^q\Lambda t_p)^{q+1}\big)$ by Eq.~\eqref{DCG_error}. A standard telescoping argument then gives $\|\widetilde{\mathcal{S}}_{2p}(T)-\mathcal{S}_{2p}(T)\| =\mathcal{O}\big(5^{p-1}l(\kappa^q\Lambda t_p)^{q+1}\big)$. Combining the two bounds proves the claim.
\end{proof}

Increasing the Trotter order $p$ suppresses the first term but amplifies the second through the factor $5^{p-1}$, leading to a trade-off governed by $\Lambda$, $T_c$, and $t_p$.

\subsection{Construction 2: Finite-width pulse-sequence input}

Consider a first-order finite-width pulse sequence as in Eq.~\eqref{Noisy_DD},
\begin{align}\label{algo2_starting_point}
    U(T_c)=\prod_{k=1}^{l} \widetilde{P}_k e^{-iH_0\tau_k},
\end{align}
satisfying the first-order condition~\eqref{finite_width_Omega1} together with the robustness conditions~\eqref{free_evolution_correction}--\eqref{pulse_width_correction}. Then $\|U(T_c)-e^{-iH_{\mathrm{targ}}T}\| = \mathcal{O}\big((\Lambda T_c)^2\big)$, with $T_c=\sum_{k=1}^l(\tau_k+t_p)$. We adapt Construction 1 to this finite-width setting.

\subsubsection{First-order Trotter mapping}
Using Eq.~\eqref{P_tilde_def}, write $\widetilde{P}_k=P_k e^{-i\Phi_k}$ with $\Phi_k=\Phi_k^{(1)}+\mathcal{O}\big((\Lambda t_p)^2\big)$. Then Eq.~\eqref{Noisy_DD} becomes
\begin{align}\label{algo2_toggling}
    U(T_c) = \prod_{k=0}^{l-1}
    e^{-ig_k^\dagger \Phi_{k+1} g_k}
    e^{-ig_k^\dagger H_0 g_k\tau_{k+1}}.
\end{align}
Using Eqs.~\eqref{free_evolution_correction}--\eqref{pulse_width_correction}, we introduce an auxiliary first-order Trotter formula $\mathcal{S}_1(T)$ with
\begin{align}\label{exponent_algo2}
    A_{2k+1} = \frac{cg_k^\dagger \Phi_{k+1}^{(1)} g_k}{T},\quad A_{2k+2} = \frac{g_k^\dagger H_0 g_k\tau_{k+1}}{T},
\end{align}
for $k=0,\dotsc,l-1$, where $c\ge 1$ is a free parameter. By construction, $A=\sum_{j=1}^{2l}A_j=H_{\mathrm{targ}}$, since $\sum_{k=0}^{l-1} g_k^\dagger \Phi_{k+1}^{(1)} g_k=0$. The corresponding first-order Trotter formula for $e^{-iH_{\mathrm{targ}}T}$ is
\begin{align}\label{S1_algo2}
    \mathcal{S}_1(T) = \prod_{k=0}^{l-1} e^{-icg_k^\dagger \Phi_{k+1}^{(1)} g_k} e^{-ig_k^\dagger H_0 g_k\tau_{k+1}}.
\end{align}

\subsubsection{Applying higher-order Trotter formulas}

We now approximate $e^{-iH_{\mathrm{targ}}T}$ using higher-order Trotter formulas for the operators $\{A_j\}$ defined in Eq.~\eqref{exponent_algo2}. Applying the second-order formula $\mathcal{S}_2(t)$ from Eq.~\eqref{Trotter_2nd} with $t=\alpha T$ gives
\begin{align}\label{S2_algo2}
\begin{aligned}
    \mathcal{S}_2(\alpha T)
    &= \prod_{k=0}^{l-1} e^{-icg_k^\dagger \Phi_{k+1}^{(1)} g_k\frac{\alpha}{2}} e^{-ig_k^\dagger H_0 g_k\frac{\alpha\tau_{k+1}}{2}} \\
    &\quad \cdot \prod_{k=l-1}^{0} e^{-ig_k^\dagger H_0 g_k\frac{\alpha\tau_{k+1}}{2}} e^{-icg_k^\dagger \Phi_{k+1}^{(1)} g_k\frac{\alpha}{2}}.
\end{aligned}
\end{align}
For $p\geq 2$, the recursion in Eq.~\eqref{Trotter_2kth} gives
\begin{align}
    \mathcal{S}_{2p}(T) = [\mathcal{S}_{2p-2}(u_pT)]^2 \mathcal{S}_{2p-2}((1-4u_p)T) [\mathcal{S}_{2p-2}(u_pT)]^2,
\end{align}
which expands into $5^{p-1}$ copies of $\mathcal{S}_2(\alpha_j T)$, where each coefficient $\alpha_j$ is a product of $p-1$ factors drawn from $\{u_k,1-4u_k\}_{k=2}^{p}$.

\subsubsection{Compiling back to robust pulse sequences}
Inverting the toggling-frame mapping in each copy of $\mathcal{S}_2(\alpha T)$ and setting $\beta:=c|\alpha|/2$, we obtain pulse factors of the form $P_k e^{\mp i\beta\Phi_k^{(1)}}$ and $e^{\mp i\beta\Phi_k^{(1)}}P_k^\dagger$, together with scaled free evolutions $e^{-iH_0\alpha\tau_k/2}$, where the upper sign applies for $\alpha>0$ and the lower sign for $\alpha<0$. The following lemma shows that these pulse factors can be implemented, up to second order in $t_p$, using only stretched and reversed versions of the input control pulses.

\begin{lemma}\label{lemma_pulse_impl}
Assume Eq.~\eqref{pulse_width_correction}. For $k=1,\dots,l$ and $\beta>0$, the following hold:
\begin{align}
    P_k e^{-i\beta\Phi_k^{(1)}}
    &= \widetilde{P}_k(\beta t_p) + \mathcal{O}\big((\Lambda\beta t_p)^2\big), \label{pulse_1}\\
    e^{-i\beta\Phi_k^{(1)}} P_k^\dagger
    &= \widetilde{P}_k^{\mathrm{rev}}(\beta t_p)
    + \mathcal{O}\big((\Lambda\beta t_p)^2\big), \label{pulse_2}\\
    e^{i\beta\Phi_k^{(1)}}P_k^\dagger
    &= \prod_{j=1}^{k-1}\widetilde{P}_j(\beta t_p) \prod_{j=k+1}^{l}\widetilde{P}_j(\beta t_p) \nonumber \\
    & \qquad + \mathcal{O}\big((l\Lambda\beta t_p)^2\big), \label{pulse_3}\\
    P_k e^{i\beta\Phi_k^{(1)}} &= \prod_{j=l}^{k+1}\widetilde{P}_j^{\mathrm{rev}}(\beta t_p) \prod_{j=k-1}^{1}\widetilde{P}_j^{\mathrm{rev}}(\beta t_p) \nonumber \\
     & \qquad + \mathcal{O}\big((l\Lambda\beta t_p)^2\big). \label{pulse_4}
\end{align}
\end{lemma}
The proof of Lemma~\ref{lemma_pulse_impl} is provided in Appendix~\ref{app:proof_lemma_pulse_impl}.

Since $t_p$ is the minimum pulse duration, all factors in Lemma~\ref{lemma_pulse_impl} must be implemented by stretching rather than compression. As each $\alpha_j$ is a product of factors drawn from $\{u_k,1-4u_k\}_{k=2}^{p}$, its smallest possible magnitude is $\prod_{k=2}^{p}u_k$. Hence the minimal such choice is
\begin{align}
    c = c_p := \frac{2}{\prod_{k=2}^{p}u_k} = 2\prod_{k=2}^{p}\bigl(4-4^{1/(2k-1)}\bigr).
\end{align}
With this choice, the worst-case stretching factor in the compiled order-$2p$ sequence is $\max_j \beta_j = \frac{c_p}{2}\max_j |\alpha_j| = 4^{\sum_{r=1}^{p-1}1/(2r+1)}$. 

\subsubsection{Error bound}

The total error again consists of the Trotter error and the finite-width implementation error from Lemma~\ref{lemma_pulse_impl}.

\begin{theorem}\label{thm_algo2}
The sequence $\widetilde{\mathcal{S}}_{2p}(T)$ obeys the bound
\begin{align}\label{alg2_bigO_clean}
\begin{aligned}
\big\|\widetilde{\mathcal{S}}_{2p}(T)-e^{-iH_{\mathrm{targ}}T}\big\|
&= \mathcal{O}\big((\Lambda(T_c+ (c_p-1) l t_p))^{2p+1}\big) \\
&\quad + \mathcal{O}\big(5^{p-1}lp^2(l\Lambda t_p)^2\big),
\end{aligned}
\end{align}
where $c_p=2\prod_{k=2}^{p}\bigl(4-4^{1/(2k-1)}\bigr)$.
\end{theorem}

\begin{proof}
By the triangle inequality,
\begin{align}
\begin{aligned}
\big\|\widetilde{\mathcal{S}}_{2p}(T)-e^{-iH_{\mathrm{targ}}T}\big\| &\le \big\|\mathcal{S}_{2p}(T)-e^{-iH_{\mathrm{targ}}T}\big\| \\
&\quad + \big\|\widetilde{\mathcal{S}}_{2p}(T)-\mathcal{S}_{2p}(T)\big\|.
\end{aligned}
\end{align}
For the first term, the decomposition~\eqref{exponent_algo2} gives $\sum_{j=1}^{2l} \|A_j\|T \le \Lambda \sum_{k=1}^{l}(c_pt_p + \tau_k) = \Lambda\bigl(T_c+(c_p-1)lt_p\bigr)$, where we used $\|\Phi_k^{(1)}\|\le \Lambda t_p$ and $T_c=\sum_{k=1}^l\tau_k+lt_p$. Eq.~\eqref{Trotter_Error} therefore yields
$\|\mathcal{S}_{2p}(T)-e^{-iH_{\mathrm{targ}}T}\| = \mathcal{O}\big((\Lambda(T_c+(c_p-1) l t_p))^{2p+1}\big)$.

For the second term, each block $\mathcal{S}_2(\alpha_jT)$ contains $2l$ pulse factors, each implemented via Lemma~\ref{lemma_pulse_impl} with error $\mathcal{O}\big((l\Lambda c_p|\alpha_j|t_p/2)^2\big)$, since $\beta_j=c_p|\alpha_j|/2$. Summing over all $5^{p-1}$ blocks and using a telescoping argument gives $\|\widetilde{\mathcal{S}}_{2p}(T)-\mathcal{S}_{2p}(T)\| =\mathcal{O}\big(5^{p-1}l \max_j \frac{c_p^2\alpha_j^2}{4}(l\Lambda t_p)^2\big)$. Now each $\alpha_j$ is a product of $p-1$ factors drawn from $\{u_k,1-4u_k\}_{k=2}^{p}$, while $c_p/2=1/\prod_{k=2}^{p}u_k$. Since $|1-4u_k|/u_k=4^{1/(2k-1)}$, the maximum of $c_p|\alpha_j|/2$ is attained when every factor is $|1-4u_k|$. Hence $\max_j c_p|\alpha_j|/2=\prod_{k=2}^{p}4^{1/(2k-1)}=\mathcal{O}(p^{\ln 2})$. Therefore $\max_j c_p^2\alpha_j^2/4=\mathcal{O}(p^{2\ln 2})$, so $\|\widetilde{\mathcal{S}}_{2p}(T)-\mathcal{S}_{2p}(T)\| =\mathcal{O}\big(5^{p-1}lp^{2\ln 2}(l\Lambda t_p)^2\big)$. In particular, this is $\mathcal{O}\big(5^{p-1}lp^2(l\Lambda t_p)^2\big)$.
Combining the two bounds proves the claim.
\end{proof}

\subsection{Implementation of negative-time evolution}\label{sec:negative_time}

For $p\ge 2$, both constructions require segments of the form $e^{+iH_0\tau}$, since the higher-order Trotter formulas contain negative coefficients. A standard way to realize such segments is to start from dynamical decoupling (DD) sequence that suppresses the effect of $H_0$. Suppose a sequence of $l_{\mathrm{neg}}$ pulses satisfies
\begin{align}\label{eq:neg_dd_identity}
    U_{\mathrm{DD}}(\tau) := \prod_{k=1}^{l_{\mathrm{neg}}} P_k e^{-iH_0\tau},
    \quad
    \|U_{\mathrm{DD}}(\tau)-I\|\le \varepsilon_{\mathrm{neg}}(\tau).
\end{align}
Then
\begin{align}\label{eq:neg_dd_block}
    U_{\mathrm{neg}}(\tau):= \left(\prod_{k=2}^{l_{\mathrm{neg}}} P_k e^{-iH_0\tau}\right) P_1
\end{align}
implements negative-time evolution with the same accuracy:
\begin{align}
    \|U_{\mathrm{neg}}(\tau)-e^{+iH_0\tau}\| = \|U_{\mathrm{DD}}(\tau)-I\|
    \le \varepsilon_{\mathrm{neg}}(\tau),
\end{align}
where we used unitary invariance of the operator norm.

For ideal pulses, concatenated DD based on a unitary $1$-design yields arbitrarily high-order suppression of $\varepsilon_{\mathrm{neg}}(\tau)$. With finite-width pulses, however, a naive implementation accumulates pulse-width errors linearly in the number of pulses, so the DD block itself must be made robust. Two natural options are available. One may replace each pulse in a concatenated DD cycle by a dynamically corrected gate. Alternatively, one may start from an Eulerian DD cycle \cite{viola2003robust}, which is already first-order robust to pulse-width effects, and then apply the concatenation ideas. In Appendix~\ref{appendix_cdd}, we analyze both constructions in detail. 


\subsection{Multi-product formulas}\label{subsec:mpf}


When the goal is to estimate expectation values rather than to prepare the full evolved state, the overhead of implementing negative-time evolution under $H_0$ can be entirely avoided by using multi-product formulas (MPFs) \cite{childs2012hamiltonian}, which combine several second-order Trotter simulations and therefore require only positive-time evolution.

For either Construction~1 or 2, write
\begin{align}
    A=\sum_{\nu=1}^{d}A_\nu, \qquad \Gamma := \sum_{\nu =1}^d \|A_\nu\|,
\end{align}
where $d=l$ for Construction~1 and $d=2l$ for Construction~2, and let $\mathcal S_2(t)$ denote the corresponding second-order Trotter formula. For pairwise distinct integers $m_1,\dots,m_p\in\mathbb N$ and coefficients $b_1,\dots,b_p\in\mathbb R$, define the MPF
\begin{align}\label{eq:mpf_def}
    M^{(p)}(T):=\sum_{j=1}^{p} b_j \bigl[\mathcal S_2(T/m_j)\bigr]^{m_j}.
\end{align}
Assume that
\begin{align}\label{eq:mpf_condition}
    \sum_{j=1}^{p} b_j = 1,
    \qquad \sum_{j=1}^{p}\frac{b_j}{m_j^{2r}}=0, \quad r=1,\dots,p-1.
\end{align}
Since $\mathcal S_2$ is symmetric, the error of $\bigl[\mathcal S_2(T/m)\bigr]^m$ admits an asymptotic expansion in even powers of $1/m$. The conditions in Eq.~\eqref{eq:mpf_condition} cancel the first $p-1$ such terms, yielding \cite{childs2012hamiltonian}
\begin{align}\label{eq:mpf_operator_error}
    \left\| M^{(p)}(T)-e^{-iAT} \right\| =
    \mathcal O\left( (\Gamma T)^{2p+1} \right).
\end{align}
In our pulse-sequence setting, we do not implement the operator-valued linear combination in Eq.~\eqref{eq:mpf_def} coherently, since that would require ancillas and controlled operations. Instead, we use the hybrid MPF protocol, in which each unitary block is executed separately and the resulting expectation values are combined classically \cite{carrera2023well}. Define
\begin{align}
\begin{aligned}
    U_j(T) &:= \bigl[\mathcal S_2(T/m_j)\bigr]^{m_j},\\
    \rho(T) &:= e^{-iAT}\rho_0e^{iAT},\\
    \rho_j(T) &:= U_j(T)\rho_0U_j^\dagger(T).
\end{aligned}
\end{align}
Then the same cancellation conditions in Eq.~\eqref{eq:mpf_condition} imply that, for any observable $O$ with $\|O\|\le 1$, \cite{carrera2023well}
\begin{align}\label{eq:mpf_expectation_error}
    \left| \Tr\bigl(O\rho(T)\bigr) - \sum_{j=1}^{p} b_j\Tr\bigl(O\rho_j(T)\bigr) \right| = \mathcal O\left( (\Gamma T)^{2p+1} \right).
\end{align}
Applying the hybrid MPF construction to the robust second-order blocks obtained from Constructions~1 and~2 yields higher-order expectation-value estimation protocols that entirely avoid negative-time evolution, since only $\mathcal S_2$ blocks are used. For each construction, the total error consists of the ideal MPF error in Eq.~\eqref{eq:mpf_expectation_error} together with the finite-width implementation error of the corresponding robust second-order pulse sequence. The resulting explicit bounds are derived in Appendix~\ref{app:mpf_bounds}.

\section{Numerical experiments}\label{sec:numerics}

We illustrate the two constructions on several benchmark pulse sequences and verify the predicted reduction in simulation error. For Construction~1, we consider two ideal-pulse sequences: one for an inhomogeneous Ising model \cite{parrarodriguez2020digitalanalog} and one for a Heisenberg model generated from cross-resonance interactions \cite{gonzalezraya2021digitalanalog}. For Construction~2, we consider a finite-width sequence for an inhomogeneous Heisenberg chain \cite{bookatz2014hamiltonian}. In all cases, we construct and benchmark the corresponding fourth-order sequences.

We quantify the error using the gate overlap
\begin{align}
    1-F(U,V) := 1-\frac{|\Tr(U^\dagger V)|}{2^n},
\end{align}
which removes the effect of global phases. This quantity satisfies $1-F(U,V) \le \frac{1}{2}\|U-V\|^2$, and therefore provides a convenient proxy for the (quadratic) simulation error. Throughout, we use rectangular pulses for simplicity, although the constructions apply to arbitrary pulse shapes that admit time stretching and reversed implementations.

\subsection{Simulating an inhomogeneous Ising model from a fixed one}
\label{sec:numerics_Ising}

We begin with a simple example that isolates the pulse-width aspect of Construction~1. Consider the homogeneous all-to-all Ising Hamiltonian,
\begin{align}\label{Ising_H0}
    H_0 = J\sum_{i<j} Z_iZ_j,
\end{align}
together with local $X$ and $Y$ controls $H_c(t)=\sum_{i=1}^n \bigl(f_i^{(X)}(t)X_i + f_i^{(Y)}(t)Y_i\bigl)$. Our goal is to engineer the target evolution generated by
\begin{align}\label{Ising_Htarg}
    H_{\mathrm{targ}}=\sum_{i<j} J_{ij} Z_iZ_j,
\end{align}
for time $T$, with random couplings $J_{ij}\in[-1,1]$.

\begin{figure}[t!]
    \centering
    \includegraphics[width=8.6cm]{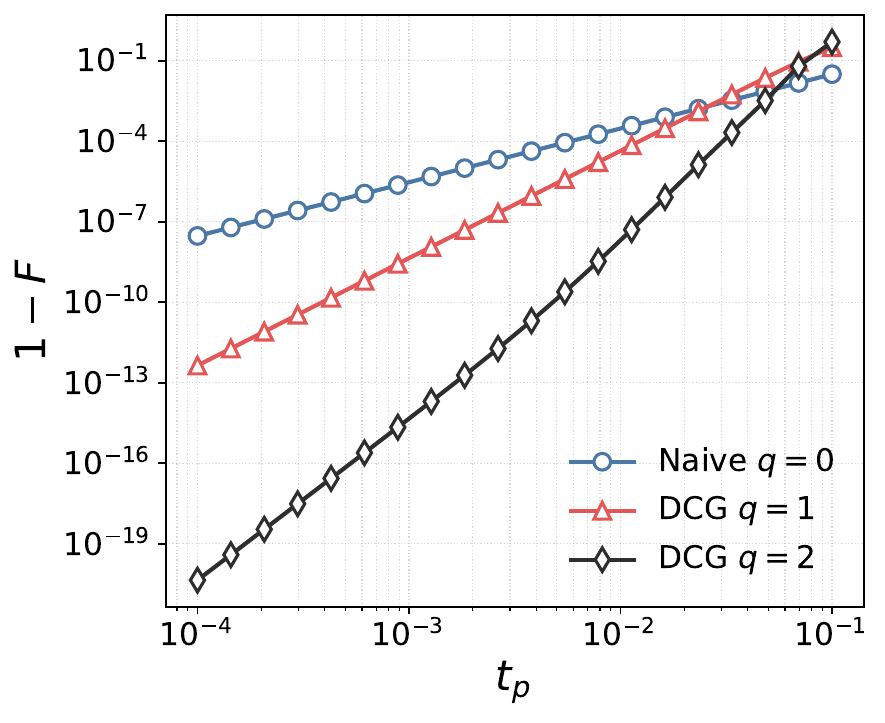}
    \caption{Simulation error versus pulse width $t_p$ for engineering the target Ising Hamiltonian in Eq.~\eqref{Ising_Htarg} from the homogeneous Ising Hamiltonian in Eq.~\eqref{Ising_H0} using the pulse sequence of Ref.~\cite{parrarodriguez2020digitalanalog}. The blue curve shows the naive finite-width implementation of the base sequence, while the red and black curves show the corresponding first- and second-order DCG implementations.}
    \label{fig:Algorithm1_Ising_to_Ising}
\end{figure}

For the $n=3$ numerical simulations, we use the pulse sequence of Ref.~\cite{parrarodriguez2020digitalanalog},
\begin{align}
    \{P_k\}=\{X_1X_2,X_2X_3,X_1X_2,X_2X_3\},
\end{align}
with interpulse delays
\begin{align}
    \{\tau_k\} =\frac{-T}{2J} \cdot \{0,J_{23}+J_{13},J_{12}+J_{23},J_{12}+J_{13}\}.
\end{align}
Some delays can be negative. In this model, however, negative-time evolution under $H_0$ can be implemented exactly by refocusing, since conjugation by local $X$ pulses only flips the signs of the commuting $ZZ$ terms. Moreover, all toggling-frame Hamiltonians commute, so the ideal sequence already realizes the target evolution exactly. Thus, for this example, Construction~1 reduces to the final DCG-replacement step for the finite-width implementation of $\{P_k,P_k^\dagger\}$. Details are given in Appendix~\ref{app:numerics_ising_dcg}.

Fig.~\ref{fig:Algorithm1_Ising_to_Ising} shows that the simulation error scales as $\mathcal O(t_p^2)$ for the naive finite-width implementation, and as $\mathcal O(t_p^4)$ and $\mathcal O(t_p^6)$ for the first- and second-order DCG implementations, respectively, in agreement with the analysis. Since the ideal sequence is exact in this commuting setting, the observed improvement arises entirely from the suppression of finite-width pulse errors.

\subsection{Simulating a Heisenberg model from cross-resonance interactions}
\label{sec:numerics_CR}

As a second example of Construction~1, we consider the cross-resonance (CR) Hamiltonian,
\begin{align}\label{system_CR}
    H_0 = J\sum_{i=1}^{n-1} X_iZ_{i+1},
\end{align}
together with local control $H_c(t) = \sum_{i=1}^n \bigl(f_i^{(Y)}(t) Y_i + f_i^{(Z)}(t) Z_i\bigr)$. Our target is to simulate the homogeneous nearest-neighbor Heisenberg Hamiltonian
\begin{align}\label{target_CR}
    H_{\mathrm{targ}} = J \sum_{i=1}^{n-1} \bigl(X_i X_{i+1} + Y_i Y_{i+1} + Z_i Z_{i+1}\bigr),
\end{align}
for time $T$.

Following Ref.~\cite{gonzalezraya2021digitalanalog}, we use the following ideal pulse sequence
\begin{align}
\begin{aligned}
    \{P_k\}_{k=1}^{4} & = \{H_{\mathrm{even}}R_E^{\dagger\otimes n}, H_{\mathrm{even}}R_E^{2\otimes n}H_{\mathrm{even}},  \\
    & \qquad H_{\mathrm{even}}R_E^{\dagger\otimes n}H_{\mathrm{even}}, H_{\mathrm{even}}\},
\end{aligned}
\end{align}
and interpulse delays
\begin{align}
    \{\tau_k\}=T\cdot \{0,1,1,1\},
\end{align}
where $H_{\mathrm{even}}$ denotes Hadamard gates applied on the even sites, and $R_E = e^{-i\frac{\pi}{3\sqrt{3}}(X+Y+Z)}$. This sequence simulates $e^{-iH_{\mathrm{targ}}T}$ with $\mathcal{O}(T_c^2)$ error where $T_c = 3T$. 

\begin{figure}[t!]
    \centering
    \includegraphics[width=8.6cm]{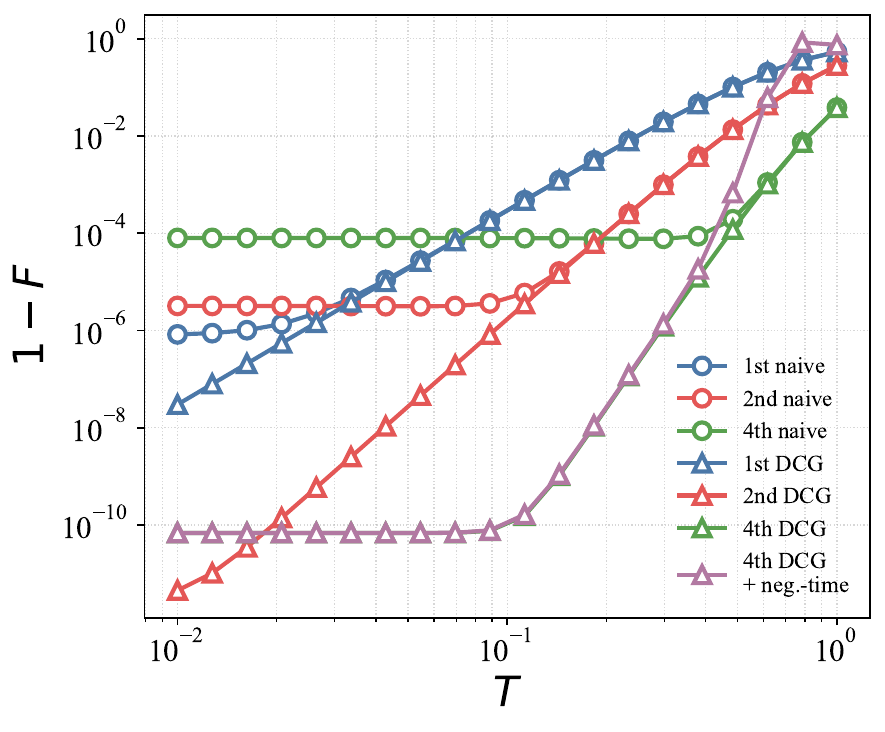}
    \caption{Simulation error versus $T$ for engineering the Heisenberg Hamiltonian in Eq.~\eqref{target_CR} from the cross-resonance Hamiltonian in Eq.~\eqref{system_CR} for $n=4,J=1$ and $t_p=10^{-4}$. Circle markers denote naive finite-width implementations of the first-, second-, and fourth-order Trotter sequences from Construction~1, while triangle markers denote the corresponding DCG-based implementations. The purple curve shows the fourth-order sequence with the negative-time segment included.}
    \label{fig:CR_result_main}
\end{figure}

We then apply Construction~1 to this base sequence to obtain robust second- and fourth-order implementations. Specifically, each ideal pulse $P_k$ and $P_k^\dagger$ is replaced by its DCG realization, while the negative-time segment required by the fourth-order Trotter formula is implemented using the concatenated Eulerian DD construction of Sec.~\ref{sec:negative_time}. The explicit pulse decompositions and the corresponding DCG implementations are given in Appendix~\ref{app:numerics_cr_DCG}.

Fig.~\ref{fig:CR_result_main} shows the simulation error for $n=4,J=1$ over $T\in[10^{-2},1]$ for $t_p=10^{-4}$. Since $T_c = 3T$, the observed scaling with respect to $T$ directly reflects the predicted high-order scaling in $T_c$. For each Trotter formula, the naive finite-width implementation exhibits a small-$T$ error plateau set by pulse-width effects, whereas the DCG-based implementation pushes this plateau to substantially smaller $T$. Away from this pulse-width-dominated regime, the observed scaling agrees with the expected order of the underlying Trotter formula. In particular, the fourth-order construction retains its predicted behavior even though its negative-time segment is implemented using finite-width pulses (see Appendix~\ref{app:numerics_cr_neg}). We also benchmark the hybrid multi-product formula construction of Sec.~\ref{subsec:mpf} for robustly estimating expectation values of observables with high accuracy (see Appendix~\ref{app:numerics_cr_mpf}).

\subsection{Simulating an anisotropic Heisenberg chain from a homogeneous one}

Our final example illustrates Construction~2, which takes as input a first-order pulse sequence that is already robust to finite pulse-width errors. We consider a one-dimensional nearest-neighbor homogeneous Heisenberg model on $n$ qubits
\begin{align}\label{num1_H0}
    H_0 = J \sum_{i=1}^{n-1} \bigl(X_iX_{i+1} + Y_iY_{i+1} + Z_iZ_{i+1}\bigr),
\end{align}
together with 1-local controls $H_c(t) = \sum_{i=1}^n \bigl(f^{(X)}_i(t) X_i + f^{(Y)}_i(t) Y_i\bigr)$. Our goal is to simulate the anisotropic Heisenberg Hamiltonian
\begin{align}\label{num1_Htarg}
    H_{\mathrm{targ}} = \sum_{i=1}^{n-1} \bigl(J_X X_iX_{i+1} + J_Y Y_iY_{i+1} + J_Z Z_iZ_{i+1}\bigr),
\end{align}
for a short time $T$ with random couplings $J_X,J_Y,J_Z \in [0,J]$.

\begin{figure}[t!]
    \centering
    \includegraphics[width=0.95\linewidth]{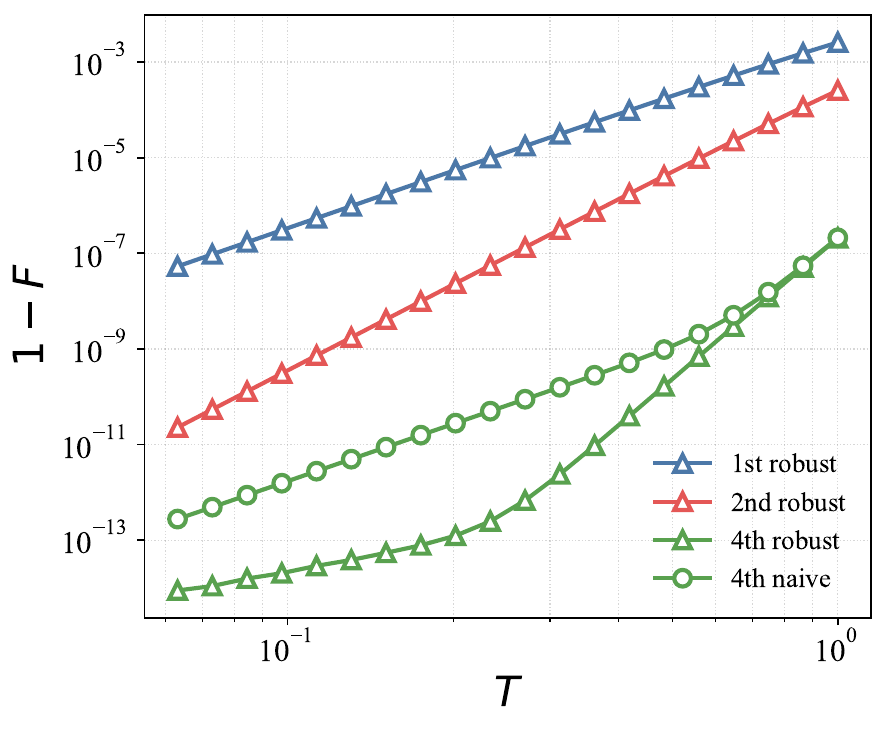}
    \caption{Simulation error for the anisotropic Heisenberg-chain example with pulse width $t_p=10^{-4}$. We compare the first-, second-, and fourth-order sequence from Construction~2, and the naive fourth-order sequence obtained without the robust compilation step of Lemma~\ref{lemma_pulse_impl}.}
    \label{fig:heisenberg_construction2}
\end{figure}

Let $\bar X:=\prod_{i~\mathrm{even}}X_i$, $\bar Y:=\prod_{i~\mathrm{even}}Y_i$, and $\bar Z:=\prod_{i~\mathrm{even}}Z_i$. Following Ref.~\cite{bookatz2014hamiltonian}, we use the pulse sequence
\begin{align}\label{num1_P}
    \{\widetilde P_k\} = \{
    \widetilde{\bar X},
    \widetilde{\bar Y},
    \widetilde{\bar X},
    \widetilde{\bar Y},
    \widetilde{\bar Y},
    \widetilde{\bar X},
    \widetilde{\bar Y},
    \widetilde{\bar X}\},
\end{align}
with interpulse delays
\begin{align}
    \{\tau_k\} = \frac{T}{4J}\cdot\{J_S,J_X,J_Z,J_Y,J_S,J_Y,J_Z,J_X\}.
\end{align}
where $J_S:=J_X+J_Y+J_Z$ with toggling operators
\begin{align}
    \{g_k\}_{k=0}^{7}=\{I,\bar X,\bar Z,\bar Y,I,\bar Y,\bar Z,\bar X\}.
\end{align}
Then Eqs.~\eqref{free_evolution_correction} and~\eqref{pulse_width_correction} are satisfied, namely
\begin{align}
    \sum_{k=0}^{7} g_k^\dagger H_0 g_k\tau_{k+1}=H_{\mathrm{targ}}T,
    \qquad
    \sum_{k=0}^{7} g_k^\dagger \Phi_{k+1}^{(1)} g_k = 0.
\end{align}
Thus, the resulting sequence simulates $e^{-iH_{\mathrm{targ}}T}$ with $\mathcal O((\Lambda T_c)^2)$ error, and therefore serves as a direct input to Construction~2. Applying Construction~2 to this robust first-order sequence yields a compiled fourth-order sequence. For comparison, we also consider a naive fourth-order sequence in which the four positive-coefficient second-order blocks are implemented as before, while the single negative-coefficient second-order block is used in its direct finite-width form without the robust replacement prescribed by Lemma~\ref{lemma_pulse_impl}.

Fig.~\ref{fig:heisenberg_construction2} compares the first-order sequence, its symmetrized second-order version, the compiled fourth-order sequence from Construction~2, and the naive fourth-order sequence at pulse width $t_p=10^{-4}$. For simplicity, we assumed that negative-time evolution under $H_0$ is directly available. Since $T_c=\frac{J_S}{J}T+8t_p$, the observed scaling versus $T$ directly reflects the predicted scaling in $T_c$ for small fixed $t_p$. As expected, the compiled fourth-order sequence gives the smallest error in the small-$T$ regime. At the same time, the naive fourth-order sequence remains surprisingly competitive, despite not satisfying the robustness condition of Lemma~\ref{lemma_pulse_impl}. This can be understood from the structure of the fourth-order formula itself: it is composed of five second-order blocks, each of which is a time-symmetrized version of an already robust first-order sequence. The underlying first-order sequence cancels the leading finite pulse-width contribution, while the subsequent symmetrization removes the leading second-order term. Consequently, even the naive fourth-order Trotter formula-based sequence inherits substantial robustness, and the additional improvement due to the robust compilation step from Lemma~\ref{lemma_pulse_impl} remains modest in this regime. We also benchmark the corresponding hybrid MPF for this example (see Appendix~\ref{app:numerics_heisenberg_mpf}).

\section{Conclusion}

We have developed a systematic method for designing fast pulse sequences that simulate target dynamics to arbitrarily high-order in the control-cycle time while remaining robust to finite pulse-width effects. The central idea is a correspondence between first-order pulse sequences and first-order Trotter formulas: by mapping a given sequence to a Trotter formula, applying high-order Trotter formulas, and translating the result back to pulses, we obtain explicit pulse sequences for robust, high-order quantum simulation. 

In particular, we present two complementary constructions. Construction~1 begins with a first-order sequence of instantaneous pulses and compiles the resulting higher-order ideal sequence into finite-width implementations using dynamically corrected gates~\cite{khodjasteh2009dynamical,khodjasteh2009dynamically,khodjasteh2010arbitrarily}, so that each control operation is robust to finite pulse-width errors. Construction~2 instead starts from a first-order sequence that is already robust to finite pulse-width effects and promotes it to higher-order while preserving the native error-cancellation structure of the original pulses, without introducing new DCGs. Numerical examples show that both constructions upgrade existing first-order protocols to the predicted higher-order scaling at realistic pulse widths.

An important caveat is that Trotter formulas beyond second-order generally require negative-time evolution under the native Hamiltonian. While such segments can be synthesized via concatenated dynamical decoupling, the resulting pulse overhead may be large, making finite-width effects significant. When only expectation values are required, however, negative-time evolution can be avoided altogether by adapting multi-product formulas~\cite{childs2012hamiltonian} to the pulse-sequence setting, combining outputs from several robust second-order sequences to obtain higher-order accuracy. For both approaches, we derive error bounds tailored to our pulse-sequence setting and validate the predicted behavior through numerical examples.

Looking ahead, the framework extends naturally to other pulse imperfections. In Construction~1, one can substitute DCGs tailored to the relevant errors~\cite{wimperis1994broadband,cummins2003tackling,brown2004arbitrarily}, while in Construction~2 the same robust implementation step applies provided the input sequence already cancels those imperfections. Fast pulses can also suppress errors due to system--environment coupling, as in dynamical decoupling \cite{viola1999dynamical}, suggesting extensions of the framework to simulation protocols robust against environmental noise. The connection to Trotter formulas also suggests incorporating randomized Trotter formulas~\cite{campbell2019random,childs2019faster} to design randomized pulse sequences, which may reduce overhead while maintaining robustness. Another promising direction is to extend these ideas from fast pulses to smooth continuous-time control, potentially informed by geometrically motivated robust-pulse constructions~\cite{zeng2018general,buterakos2021geometrical,barnes2022dynamically}.

\section*{Acknowledgments}
This work is supported by DOE’s Express 2023 Number DE-SC0024685.

\section*{Data Availability}
The data that support the findings of this article are openly available \cite{kim2026robust}.

\appendix

\section{Dynamically corrected gates}\label{app:dcg}
Here we summarize the dynamically corrected gate (DCG) construction used in Sec.~\ref{subsec:construction1}, specialized to the present setting where the only pulse error is the always-on native Hamiltonian $H_0$. We follow Refs.~\cite{khodjasteh2009dynamical,khodjasteh2009dynamically,khodjasteh2010arbitrarily} and use the pulse notation of Eqs.~\eqref{P_tilde_def}, \eqref{magnus_1}, \eqref{eq:stretched_pulse}, and \eqref{eq:stretched_rev_pulse}.

Let $W$ be an ideal pulse generated by a control function $f_W(t)$ supported on $[0,t_p]$. Defining $H_W$ for the corresponding control Hamiltonian, the ideal pulse propagator is
\begin{align}
    U_W(t)=\mathcal{T}\exp\left(-i\int_0^t f_W(t')H_Wdt'\right), \quad
    U_W(t_p)=W.
\end{align}
Its finite-width implementation is
\begin{align}
    \widetilde W := \mathcal{T}\exp\left(-i\int_0^{t_p}\big(H_0+f_W(t)H_W\big) dt \right)
    = W e^{-i\Phi_W},
\end{align}
where $\Phi_W$ is the Magnus expansion. Its first-order term is
\begin{align}
    \Phi_W^{(1)}=\int_0^{t_p}U_W^\dagger(t)H_0U_W(t)dt,
\end{align}
so $\|\Phi_W\|=\mathcal{O}(\Lambda t_p)$.

\subsection{First-order Eulerian DCGs}

Let $\mathcal E$ be a linear space of error operators. A finite subgroup $\mathcal G\subset U(\mathcal H)$ is called a decoupling group for $\mathcal E$ if
\begin{align}
    \frac{1}{|\mathcal G|}\sum_{g\in\mathcal G} g^\dagger E g = 0,
    \qquad \forall E\in\mathcal E.
    \label{eq:dcg_decoupling_condition}
\end{align}
Let $\Gamma=\{h_1,\dots,h_m\}$ be a generating set for $\mathcal G$, and let $\mathrm{Cay}(\mathcal G,\Gamma)$ be the corresponding directed Cayley graph, whose vertices are the elements of $\mathcal G$ and whose directed edges are $g\to gh$. An Eulerian cycle on this graph traverses each directed edge exactly once.

Assume that each generator $h\in\Gamma$ is always implemented by the same finite-width control function,
\begin{align}
    \widetilde h = h e^{-i\Phi_h}.
\end{align}
Let $\Omega_\Gamma$ denote the Magnus expansion of the composite propagator obtained by following one Eulerian cycle on $\mathrm{Cay}(\mathcal G,\Gamma)$ using these pulse implementations. Since each directed edge $g \to gh$ is traversed exactly once, for each generator $h\in\Gamma$ the corresponding first-order error $\Phi_h^{(1)}$ appears once in every group frame $g\in\mathcal G$. Therefore,
\begin{align}
    \Omega_\Gamma^{(1)} = \sum_{h\in\Gamma}\sum_{g\in\mathcal G} g^\dagger \Phi_h^{(1)} g.
    \label{eq:dcg_generator_average}
\end{align}
Hence $\Omega_\Gamma^{(1)}=0$ whenever $\Phi_h^{(1)}\in\mathcal E$ for all $h\in\Gamma$ and $\mathcal G$ satisfies the decoupling condition in Eq.~\eqref{eq:dcg_decoupling_condition}.

\begin{figure*}[t!]
    \centering
    \includegraphics[width=17cm]{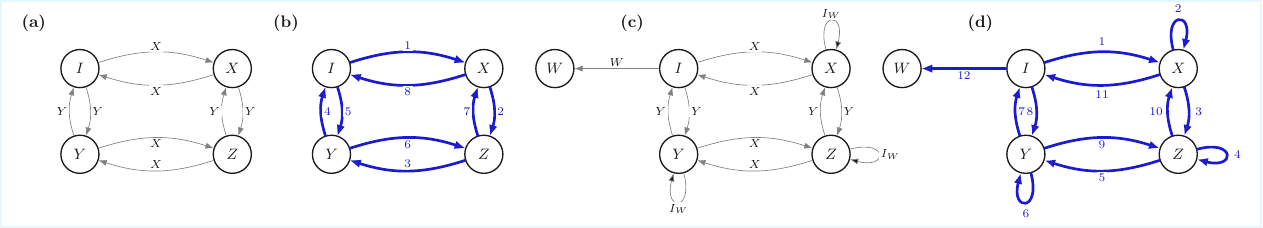}
    \caption{Illustration of the first-order Eulerian DCG construction.  (a) Directed Cayley graph $\mathrm{Cay}(\mathcal G,\Gamma)$ for $\mathcal G=\{I,X,Y,Z\}$ and $\Gamma=\{X,Y\}$.  (b) An Eulerian cycle generating the EDD identity block.  (c) The augmented graph of (a) with $I_W$ self-loops and a target edge from the identity vertex. (d) An Eulerian path on the augmented graph, yielding a first-order dynamically corrected implementation of $W$.}
    \label{fig:placeholder}
\end{figure*}

Eulerian averaging therefore cancels the first-order errors of the generator pulses, but it produces only a robust implementation of the identity: an Eulerian cycle starts and ends at the identity vertex. This identity block is the basic Eulerian dynamical decoupling (EDD) construction \cite{viola2003robust}. Figs.~\ref{fig:placeholder}(a) and \ref{fig:placeholder}(b) illustrate the directed Cayley graph and one corresponding Eulerian cycle for $\mathcal G=\{I,X,Y,Z\}$ and $\Gamma=\{X,Y\}$. 

To implement a nontrivial target gate $W$, we use the reversed and stretched versions of the same control profile. The reversed implementation is
\begin{align}
    \widetilde W^{\mathrm{rev}} :=\mathcal{T}\exp\left(-i\int_0^{t_p}\big(H_0-f_W(t_p-t)H_W\big)dt\right),
\end{align}
which implements $W^\dagger$ in the absence of $H_0$. For $c\ge 1$, the stretched implementation is
\begin{align}
    \widetilde W(ct_p) :=\mathcal{T}\exp\left(-i\int_0^{ct_p}\Big(H_0+\frac1c f_W(t/c)H_W\Big)dt\right),
\end{align}
which implements $W$ in the absence of $H_0$. Using these ingredients, we construct a balance pair, consisting of a target block and an identity block with the same leading error.

For the first-order construction, a convenient choice is
\begin{align}\label{I_W}
    \widetilde I_W := \widetilde W^{\mathrm{rev}}\widetilde W
\end{align}
together with the stretched pulse $\widetilde W(2t_p)$. As in the derivation of Eq.~\eqref{pulse_2},
\begin{align}
    \widetilde I_W = e^{-i2\Phi_W^{(1)}}+\mathcal{O}\big((\Lambda t_p)^2\big).
\end{align}
By Proposition~\ref{prop:uniform_stretching}, the first-order Magnus term of the stretched pulse $\widetilde W(2t_p)$ is $2\Phi_W^{(1)}$. Hence $\widetilde I_W$ and $\widetilde W(2t_p)$ share the same first-order error generator.

The first-order Eulerian DCG is obtained by augmenting $\mathrm{Cay}(\mathcal G,\Gamma)$ as follows: attach a self-loop labeled $\widetilde I_W$ at each non-identity vertex $g\neq I$, and replace the loop at the identity by an outgoing edge labeled $\widetilde W(2t_p)$ leading to a terminal vertex. An Eulerian path on this augmented graph starts at $I$, traverses each original edge and each added self-loop exactly once, and exits through the edge labeled $\widetilde W(2t_p)$. The resulting composite operation implements $W$ ideally. Figs.~\ref{fig:placeholder}(c) and \ref{fig:placeholder}(d) show the corresponding augmented graph and an Eulerian path implementing the first-order DCG construction.

Let $\Omega_{\mathrm{DCG}}$ denote the Magnus expansion of the composite pulse sequences associated with this augmented-graph sequence. The first-order term of $\Omega_{\mathrm{DCG}}$ is
\begin{align}
    \Omega_{\mathrm{DCG}}^{(1)} = \sum_{h\in\Gamma}\sum_{g\in\mathcal G} g^\dagger \Phi_h^{(1)} g + \sum_{g\in\mathcal G} g^\dagger \bigl(2\Phi_W^{(1)}\bigr) g.
    \label{eq:dcg_first_order_error}
\end{align}
If $\Phi_h^{(1)}\in\mathcal E$ for all $h\in\Gamma$ and $2\Phi_W^{(1)}\in\mathcal E$, then Eq.~\eqref{eq:dcg_decoupling_condition} implies $\Omega_{\mathrm{DCG}}^{(1)}=0$. The leading nonvanishing Magnus term is therefore at least second order in $\Lambda\tau_1$, where $\tau_1$ is the total duration of the DCG block. Hence the resulting first-order DCG implementation, denoted $\widetilde W^{[1]}$, satisfies
\begin{align}
    \|\widetilde W^{[1]}-W\|
    = \mathcal{O}\big((\Lambda\tau_1)^2\big).
\end{align}

\subsection{Higher-order concatenated Eulerian DCGs} 
The first-order Eulerian construction extends recursively to arbitrarily high
order. Following Ref.~\cite{khodjasteh2010arbitrarily}, suppose that for some
$q\ge 1$ we have already constructed a $q$th-order implementation of $W$,
denoted $\widetilde W^{[q]}$, with total duration $\tau_q$, such that
\begin{align}
    \widetilde W^{[q]} = W e^{-i\Phi_W^{[q]}},
    \qquad
    \|\Phi_W^{[q]}\| = \mathcal{O}\big((\Lambda\tau_q)^{q+1}\big).
    \label{eq:dcg_q_error}
\end{align}
Likewise, let $\widetilde{W^\dagger}^{[q]}$ denote a $q$th-order
implementation of $W^\dagger$.

Given a $q$th-order implementation $\widetilde W^{[q]}$, we write $\widetilde W^{[q]}(r\tau_q)$ for the uniformly stretched version of its control profile by a factor $r\ge1$, in the same sense as the stretched pulses defined in Sec.~\ref{subsec:control-function-variations}. Let
\begin{align}\label{cdcg-balance-pair-con}
    r_q:=2^{1/(q+1)}.
\end{align}
As in the first-order case, the recursive construction is based on a $q$th-order balance pair: an identity block $\widetilde I_W^{[q]}$ and a target block $\widetilde W_*^{[q]}$ whose residual errors agree up to order $q+1$. One choice is \cite{khodjasteh2010arbitrarily}
\begin{align}
    \widetilde I_W^{[q]} &:= \widetilde{W^\dagger}^{[q]}(\tau_q) \widetilde W^{[q]}\big(r_q\tau_q\big),\label{cDCG_balance_pair_I}
    \\
    \widetilde W_*^{[q]} &:= \widetilde W^{[q]}(\tau_q) \widetilde{W^\dagger}^{[q]}(\tau_q) \widetilde W^{[q]}(\tau_q).
    \label{cDCG_balance_pair_W}
\end{align}
The first block implements the identity, while the second implements $W$, and
their residual errors satisfy
\begin{align}
    \widetilde I_W^{[q]}
    &=I e^{-i\Phi_{I_W}^{[q]}}, \label{cDCG_balance_pair_I_err}
    \\
    \widetilde W_*^{[q]}
    &=W e^{-i\Phi_{W_*}^{[q]}}.\label{cDCG_balance_pair_W_err}
\end{align}
Ref.~\cite{khodjasteh2010arbitrarily} states, without explicit proof, that the residual errors agree up to order $q+1$. For completeness, we verify this matching at the end of this subsection (see Lemma~\ref{lem:dcg_balance_pair_matching}).

One then repeats the Eulerian construction of the previous subsection at the level of these $q$th-order blocks: choose a decoupling group for the error model generated by the leading residual errors of the $q$th-order ingredients, form the corresponding augmented Cayley graph, and define $\widetilde W^{[q+1]}$ by following an Eulerian path on that graph. Exactly as before, each leading error generator is averaged once over every group frame, so all contributions of order $(\Lambda\tau_q)^{q+1}$ cancel. The resulting gate therefore satisfies 
\begin{align} 
\widetilde W^{[q+1]} = W e^{-i\Phi_W^{[q+1]}}, \quad \|\Phi_W^{[q+1]}\| = \mathcal{O}\big((\Lambda\tau_{q+1})^{q+2}\big), 
\end{align} 
where $\tau_{q+1}$ is the total duration of the concatenated block. (Note that a $q$th-order concatenated DCG requires pulse stretching by at most a factor $2^{\sum_{r=1}^{q}1/r}$.) Ref.~\cite{khodjasteh2010arbitrarily} gives a recursive growth law for the total duration under concatenation
\begin{align}
    \tau_{q+1} = \Big[d_qm_q+(d_q-1)\big(1+2^{1/(q+1)}\big)+3\Big]\tau_q,
\end{align}
where $d_q:=|G^{[q]}|$ is the order of the decoupling group at level $q$, and $m_q$ is the number of generators in its generating set. The bracketed factor is the per-level overhead incurred by the augmented Eulerian construction. If these overhead factors are uniformly bounded above by some construction-dependent constant $\kappa>1$, then iterating the recursion gives
\begin{align}
    \tau_q \le \kappa^q t_p .
\end{align}
Combining this with Eq.~\eqref{eq:dcg_q_error} gives 
\begin{align} \label{DCG_error_apdx}
\|\widetilde W^{[q]}-W\| = \mathcal{O}\big((\kappa^q\Lambda t_p)^{q+1}\big), 
\end{align} 
which is the estimate used in Eq.~\eqref{DCG_error}.

It remains to verify the matching property of the balance pair. We first note
a simple stretching identity.

\begin{proposition}\label{prop:uniform_stretching}
Let $\widetilde W(T)=We^{-i\Phi_W(T)}$ be an implementation of $W$ generated by a control Hamiltonian $H_c(t)$ on $[0,T]$. For $r\ge1$, define
\begin{align}
    H_c^{(r)}(t)=r^{-1}H_c(t/r), \qquad t\in[0,rT],
\end{align}
and let $\widetilde W(rT)=We^{-i\Phi_W(rT)}$ be the corresponding stretched implementation. Defining $\Phi_W(s)=\sum_{m\ge1}\Phi_W^{(m)}(s)$ for the Magnus expansion of the error action, for $s = T$ and $s = rT$, one has
\begin{align}
    \Phi_W^{(m)}(rT)=r^m\Phi_W^{(m)}(T), \qquad m\ge1.
\end{align}
\end{proposition}

\begin{proof}
Let
\begin{align}
    U_c^{(r)}(t) = \mathcal T\exp\left(-i\int_0^t r^{-1}H_c(t'/r)dt'\right)
\end{align}
be the control propagator generated by $H_c^{(r)}(t)$. With the change of variables $s=t'/r$, we obtain
\begin{align}
    U_c^{(r)}(t) =\mathcal T\exp\left(-i\int_0^{t/r} H_c(s)ds\right) =U_c(t/r)
\end{align}
and hence $U_c^{(r)}(rT)=U_c(T)=W$. Therefore the toggling-frame Hamiltonian for the stretched implementation is
\begin{align}
    H_I^{(r)}(t) = \bigl(U_c^{(r)}(t)\bigr)^\dagger H_0 U_c^{(r)}(t) = H_I(t/r).
\end{align}
Now the $m$th Magnus term can be written as
\begin{align}
    &\Phi_W^{(m)}(rT) = \\
    & \int_{0\le t_m\le\cdots\le t_1\le rT}
    F_m\bigl(H_I^{(r)}(t_1),\dots,H_I^{(r)}(t_m)\bigr)
    dt_1\cdots dt_m, \nonumber 
\end{align}
where $F_m$ is the same multilinear nested-commutator expression as in the unstretched case. Using $H_I^{(r)}(t_j)=H_I(t_j/r)$ and changing variables $t_j=rs_j$ for $j=1,\dots,m$, we obtain
\begin{align}
    & \Phi_W^{(m)}(rT)
    = \\
    & r^m \int_{0\le s_m\le\cdots\le s_1\le T} F_m\bigl(H_I(s_1),\dots,H_I(s_m)\bigr)
    ds_1\cdots ds_m \nonumber,
\end{align} 
which is equal to $r^m \Phi_W^{(m)}(T)$.
\end{proof}

We now show that the balance pair defined above is indeed matched up to order $q+1$.

\begin{lemma}\label{lem:dcg_balance_pair_matching}
The balance pair defined in Eqs.~\eqref{cDCG_balance_pair_I} and~\eqref{cDCG_balance_pair_W} satisfies
\begin{align}
    \Phi_{I_W}^{[q]}-\Phi_{W_*}^{[q]} = \mathcal{O}\big((\Lambda\tau_q)^{q+2}\big).
\end{align}
\end{lemma}

\begin{proof}
Since $\widetilde W^{[q]}$ and $\widetilde{W^\dagger}^{[q]}$ are $q$th-order implementations, their error actions have the expansions
\begin{align}
    \Phi_W^{[q]}(\tau_q) &= A_W^{[q]}\tau_q^{q+1} + \mathcal{O}\big((\Lambda\tau_q)^{q+2}\big), \\
    \Phi_{W^\dagger}^{[q]}(\tau_q) &= A_{W^\dagger}^{[q]}\tau_q^{q+1}  + \mathcal{O}\big((\Lambda\tau_q)^{q+2}\big).
\end{align}
By Proposition~\ref{prop:uniform_stretching},
\begin{align}
    \Phi_W^{[q]}(r_q\tau_q) = r_q^{q+1}A_W^{[q]}\tau_q^{q+1}  + \mathcal{O}\big((\Lambda\tau_q)^{q+2}\big).
\end{align}
For $\widetilde A=Ae^{-i\Phi_A}$ and $\widetilde B=Be^{-i\Phi_B}$, let $\Phi_{AB}$ denote the error action of the composite implementation, defined by $\widetilde A\widetilde B = AB e^{-i\Phi_{AB}}$. Then the Baker--Campbell--Hausdorff (BCH) formula gives
\begin{align}
    \Phi_{AB} = \Phi_B + B^\dagger \Phi_A B + \mathcal{O}\big(\|\Phi_A\|\|\Phi_B\|\big).
\end{align}
Here $\Phi_A,\Phi_B=\mathcal{O}\big((\Lambda\tau_q)^{q+1}\big)$, so the BCH remainder is $\mathcal{O}\big((\Lambda\tau_q)^{2q+2}\big) = \mathcal{O}\big((\Lambda\tau_q)^{q+2}\big)$.

Applying this to the two elements of the balance pair yields
\begin{align}
    \Phi_{I_W}^{[q]}
    &= \Big(r_q^{q+1}A_W^{[q]} + W^\dagger A_{W^\dagger}^{[q]}W \Big)\tau_q^{q+1} + \mathcal{O}\big((\Lambda\tau_q)^{q+2}\big), \\
    \Phi_{W_*}^{[q]} &=
    \Big(2A_W^{[q]} + W^\dagger A_{W^\dagger}^{[q]}W \Big)\tau_q^{q+1} + \mathcal{O}\big((\Lambda\tau_q)^{q+2}\big).
\end{align}
Since $r_q^{q+1}=2$ from Eq.~\eqref{cdcg-balance-pair-con}, the order-$(q+1)$ terms coincide, and therefore
\begin{align}
    \Phi_{I_W}^{[q]}-\Phi_{W_*}^{[q]}
    = \mathcal{O}\big((\Lambda\tau_q)^{q+2}\big).
\end{align}
\end{proof}

\section{Proof of Lemma~\ref{lemma_pulse_impl}}
\label{app:proof_lemma_pulse_impl}

\begin{proof}[Proof of Lemma~\ref{lemma_pulse_impl}]
By Proposition~\ref{prop:uniform_stretching}, it suffices to prove the four identities for $\beta=1$, since the general case then follows by replacing $t_p$ with $\beta t_p$.

\paragraph{Proof of Eq.~\eqref{pulse_1}.}
Using $\widetilde{P}_k=P_k e^{-i\Phi_k}$ together with
$\Phi_k=\Phi_k^{(1)}+\mathcal{O}\big((\Lambda t_p)^2\big)$, we obtain
\begin{align}
\|P_k e^{-i\Phi_k^{(1)}}-\widetilde{P}_k\| = \mathcal{O}\big((\Lambda t_p)^2\big).
\end{align}

\paragraph{Proof of Eq.~\eqref{pulse_2}.}
Consider the composite operation $\widetilde{P}_k^{\mathrm{rev}}\widetilde{P}_k$, generated by the control function
\begin{align}
f'_k(t)=
\begin{cases}
f_k(t), & t\in[0,t_p],\\
-f_k(2t_p-t), & t\in(t_p,2t_p].
\end{cases}
\end{align}
Since $\int_0^{2t_p} f'_k(t)dt=0$, the corresponding ideal pulse action is trivial, and thus
\begin{align}
\widetilde{P}_k^{\mathrm{rev}}\widetilde{P}_k=e^{-i\Phi'_k}
\end{align}
in the interaction picture for some Magnus expansion $\Phi'_k$. The associated control propagator satisfies
\begin{align}
U_c'(t)=U_{P_k}(2t_p-t), \qquad t\in(t_p,2t_p],
\end{align}
so $U_c'(t)=U_c'(2t_p-t)$ on $[0,2t_p]$. Hence
\begin{align}
{\Phi'_k}^{(1)}
&= \int_0^{2t_p} {U_c'}^\dagger(t) H_0 U_c'(t)dt \\
&= 2\int_0^{t_p} U_{P_k}^\dagger(t) H_0 U_{P_k}(t)dt = 2\Phi_k^{(1)}.
\end{align}
Therefore,
\begin{align}
\|\widetilde{P}_k^{\mathrm{rev}}\widetilde{P}_k-e^{-i2\Phi_k^{(1)}}\| = \mathcal{O}\big((\Lambda t_p)^2\big).
\end{align}
Right-multiplying by $\widetilde{P}_k^\dagger=e^{i\Phi_k}P_k^\dagger$ and using $\Phi_k=\Phi_k^{(1)}+\mathcal{O}\big((\Lambda t_p)^2\big)$ gives
\begin{align}
\|e^{-i\Phi_k^{(1)}}P_k^\dagger-\widetilde{P}_k^{\mathrm{rev}}\| = \mathcal{O}\big((\Lambda t_p)^2\big),
\end{align}
which proves Eq.~\eqref{pulse_2}.

\paragraph{Proof of Eq.~\eqref{pulse_3}.}
By Eq.~\eqref{pulse_width_correction},
\begin{align}
\sum_{k=0}^{l-1} g_k^\dagger \Phi_{k+1}^{(1)} g_k = 0.
\end{align}
Applying the first-order Trotter formula to the vanishing sum gives
\begin{align}\label{full_cycle_id_appendix}
\left\|I-\prod_{k=1}^{l}P_k e^{-i\Phi_k^{(1)}}\right\| = \mathcal{O}\big((l\Lambda t_p)^2\big).
\end{align}
We now isolate the $k$th factor. Since the operator norm is unitarily invariant, and both $P_j$ and $e^{-i\Phi_j^{(1)}}$ are unitary (all Magnus terms are Hermitian \cite{magnus1954on}), the bound is preserved under the following left and right multiplications. Left-multiplying by $e^{i\Phi_j^{(1)}}P_j^\dagger$ for $j=l,l-1,\dotsc,k$ and right-multiplying by $P_j e^{-i\Phi_j^{(1)}}$ for $j=l,l-1,\dotsc,k+1$ yields
\begin{align}
\left\|e^{i\Phi_k^{(1)}}P_k^\dagger - \prod_{j=1}^{k-1}P_j e^{-i\Phi_j^{(1)}} \prod_{j=k+1}^{l}P_j e^{-i\Phi_j^{(1)}}\right\| = \mathcal{O}\big((l\Lambda t_p)^2\big).
\end{align}
Substituting Eq.~\eqref{pulse_1} into the remaining factors gives
\begin{align}
e^{i\Phi_k^{(1)}}P_k^\dagger = \prod_{j=1}^{k-1}\widetilde{P}_j(t_p) \prod_{j=k+1}^{l}\widetilde{P}_j(t_p) + \mathcal{O}\big((l\Lambda t_p)^2\big),
\end{align}
which is Eq.~\eqref{pulse_3} for $\beta=1$.

\paragraph{Proof of Eq.~\eqref{pulse_4}.}
Applying the first-order Trotter formula to
\begin{align}
\sum_{k=0}^{l-1} g_k^\dagger (-\Phi_{k+1}^{(1)}) g_k = 0
\end{align}
gives
\begin{align}
\left\|I-\prod_{k=1}^{l}P_k e^{i\Phi_k^{(1)}}\right\| = \mathcal{O}\big((l\Lambda t_p)^2\big).
\end{align}
Using the same isolation procedure as above, we obtain
\begin{align}
\left\|P_k e^{i\Phi_k^{(1)}} - \prod_{j=l}^{k+1}e^{-i\Phi_j^{(1)}}P_j^\dagger \prod_{j=k-1}^{1}e^{-i\Phi_j^{(1)}}P_j^\dagger\right\| = \mathcal{O}\big((l\Lambda t_p)^2\big).
\end{align}
Substituting Eq.~\eqref{pulse_2} for the remaining factors yields
\begin{align}
P_k e^{i\Phi_k^{(1)}} = \prod_{j=l}^{k+1}\widetilde{P}_j^{\mathrm{rev}}(t_p) \prod_{j=k-1}^{1}\widetilde{P}_j^{\mathrm{rev}}(t_p) + \mathcal{O}\big((l\Lambda t_p)^2\big),
\end{align}
which is Eq.~\eqref{pulse_4} for $\beta=1$.
\end{proof}

\section{Implementation of negative-time evolution via concatenated DD}
\label{appendix_cdd}

We summarize the robust implementation of negative-time evolution described in Sec.~\ref{sec:negative_time}. As noted in Eqs.~\eqref{eq:neg_dd_identity} and~\eqref{eq:neg_dd_block}, the problem reduces to constructing a dynamical decoupling sequence that approximates the identity while suppressing the evolution generated by $H_0$. Such a sequence immediately yields an approximation of $e^{+iH_0\tau}$ with the same error.

\subsection{Symmetrized CDD with DCGs}\label{subsec:appendix_neg_time_dcg}

We first consider a finite-width robust version of concatenated dynamical decoupling (CDD). Here a first-order DD cycle is a pulse sequence that averages the native Hamiltonian $H_0$ to zero to leading order, so that the resulting evolution approximates the identity \cite{viola1999dynamical}. Concatenation recursively nests such cycles to achieve higher-order suppression \cite{khodjasteh2005fault}, and time symmetrization removes the even-order Magnus terms \cite{iserles2001time}.

Let $U_{\rm sym}^{(0)}(\tau):=e^{-iH_0\tau}$, and for $k\ge 1$ define recursively
\begin{align}
    U_{\rm sym}^{(k)}(\tau) := \prod_{j=1}^{l} P_j U_{\rm sym}^{(k-1)}(\tau)
    \prod_{j=l}^{1} U_{\rm sym}^{(k-1)}(\tau) P_j^\dagger.
    \label{eq:sym_cdd_recursive_app}
\end{align}
Thus the level-$k$ block is obtained by replacing each free-evolution segment in the time-symmetrized base cycle by a copy of the level-$(k-1)$ block.

Let $g_j:=P_j\cdots P_1$ for $j=1,\dots,l$, with $g_0:=I$, and for simplicity assume that $\{g_j\}_{j=0}^{l-1}$ forms a unitary $1$-design, i.e.,
\begin{align}
    \frac1l \sum_{j=0}^{l-1} g_j^\dagger A g_j = 0
\end{align}
for every traceless operator $A$. To analyze the recursion, write
\begin{align}
    U_{\rm sym}^{(k-1)}(\tau)=e^{-iE_{k-1}(\tau)},
\end{align}
where $E_{k-1}(\tau)$ is the residual error action of the level-$(k-1)$ block. Then Eq.~\eqref{eq:sym_cdd_recursive_app} is a palindromic product of $2l$ segments generated by the toggled operators $g_j^\dagger E_{k-1}(\tau) g_j$. Since $E_{k-1}(\tau)$ is traceless, the $1$-design property implies that the first Magnus term of the outer cycle vanishes. Moreover, the palindromic structure implies that the toggling-frame Hamiltonian for the level-$k$ cycle is time-symmetric, so all even-order Magnus terms vanish. Therefore the leading error contribution is 
\begin{align}
    \|E_k(\tau)\| = \mathcal{O}\big(l^3\|E_{k-1}(\tau)\|^3\big),
    \label{eq:sym_cdd_error_recursion_app}
\end{align}
assuming Magnus convergence for each level.

Since $\|E_0(\tau)\|=\Lambda\tau$, iterating Eq.~\eqref{eq:sym_cdd_error_recursion_app} yields
\begin{align}
    \|U_{\rm sym}^{(k)}(\tau)-I\| = \mathcal{O}\big(l^{\frac{3^{k+1}-3}{2}}(\Lambda\tau)^{3^k}\big).
    \label{eq:ideal_sym_cdd_bound_app}
\end{align}
We now make this construction robust to finite pulse widths by replacing each ideal pulse $W\in\{P_j,P_j^\dagger\}$ by a $q$th-order DCG implementation $\widetilde W^{[q]}$, discussed in Appendix~\ref{app:dcg}. As in Eq.~\eqref{DCG_error_apdx}, these satisfy
\begin{align}
    \|\widetilde W^{[q]}-W\| = \mathcal{O}\big((\kappa^q\Lambda t_p)^{q+1}\big).
    \label{eq:neg_time_dcg_single}
\end{align}
Let $\widetilde U_{\rm sym}^{(k)}(\tau)$ denote the sequence obtained from $U_{\rm sym}^{(k)}(\tau)$ by replacing every ideal pulse by its DCG implementation. Since the number of pulse occurrences in the level-$k$ block scales as $\mathcal O((2l)^k)$, a telescoping expansion gives
\begin{align}
\begin{aligned}
    \|\widetilde U_{\rm sym}^{(k)}(\tau)-I\|
    &= \mathcal{O}\big(l^{\frac{3^{k+1}-3}{2}}(\Lambda\tau)^{3^k}\big) \\
    &\quad+\mathcal{O}\big((2l)^k(\kappa^q\Lambda t_p)^{q+1}\big).
\end{aligned}
\end{align}
For a generic $n$-qubit Hamiltonian $H_0$, a DD cycle typically has size exponential in $n$. For example, the $n$-qubit Pauli group gives $l=4^n$. Thus increasing the concatenation level improves the ideal suppression only up to the point where the accumulated DCG implementation error becomes comparable. 

\subsection{Symmetrized concatenated Eulerian DD} \label{subsec:appendix_neg_time_eulerian}

An alternative robust implementation is to build the finite-width robustness directly into the decoupling cycle, rather than into each pulse separately. This is achieved by starting from an Eulerian dynamical decoupling cycle \cite{viola2003robust}, which cancels both the free-evolution contribution and the leading pulse-width error to first order. As reviewed in Appendix~\ref{app:dcg}, EDD realizes the decoupling group by following an Eulerian cycle on its Cayley graph. Eq.~\eqref{eq:dcg_generator_average} then shows that the first Magnus term vanishes whenever Eq.~\eqref{eq:dcg_decoupling_condition} is satisfied.

Let
\begin{align}
    U_{\rm EDD}^{(1)}(\tau) := \prod_{j=1}^{l_{\rm EDD}} \widetilde P_j e^{-iH_0\tau}
\end{align}
be a first-order robust EDD cycle of length $l_{\rm EDD}$, and let
\begin{align}
    U_{\rm sym,EDD}^{(1)}(\tau) := \prod_{j=1}^{l_{\rm EDD}} \widetilde P_j e^{-iH_0\tau} \prod_{j=l_{\rm EDD}}^{1} e^{-iH_0\tau}\widetilde P_j^{\rm rev}
    \label{eq:sym_edd_base_app}
\end{align}
be its time-symmetrized version, where $\widetilde P_j^{\rm rev}$ denotes the reversed implementation of $\widetilde P_j$ defined in Sec.~\ref{subsec:control-function-variations}. Since the Eulerian cycle cancels the first Magnus term, including both the free-evolution contribution from $H_0$ and the leading pulse-width error, it remains to show that symmetrization removes the second Magnus term. Let $T_{\rm sym}:=2l_{\rm EDD}(\tau+t_p)$, and let $H_c(t)$ and $U_c(t)$ denote the control Hamiltonian and control propagator of the symmetrized cycle. By construction of the reversed pulses, $H_c(T_{\rm sym}-t)=-H_c(t)$. Since the cycle is cyclic, $U_c(T_{\rm sym})=I$. Hence $U_c(T_{\rm sym}-t)$ satisfies the same Schr\"odinger equation and initial condition as $U_c(t)$, so $U_c(T_{\rm sym}-t)=U_c(t)$. Therefore the corresponding toggling-frame Hamiltonian $H_I(t):=U_c^\dagger(t)H_0U_c(t)$ is time-symmetric, $H_I(T_{\rm sym}-t)=H_I(t)$, and thus all even-order Magnus terms vanish \cite{iserles2001time}. Hence,
\begin{align}
    \|U_{\rm sym,EDD}^{(1)}(\tau)-I\| = \mathcal O\big(l_{\rm EDD}^3[\Lambda(\tau+t_p)]^3\big).
    \label{eq:sym_edd_base_bound_app}
\end{align}
Let $U_{\rm sym,EDD}^{(0)}(\tau):=e^{-iH_0\tau}$ and for $k\ge 1$ define recursively
\begin{align}
    U_{\rm sym,EDD}^{(k)}(\tau):= \prod_{j=1}^{l_{\rm EDD}} \widetilde P_j U_{\rm sym,EDD}^{(k-1)}(\tau) \prod_{j=l_{\rm EDD}}^{1} U_{\rm sym,EDD}^{(k-1)}(\tau)\widetilde P_j^{\rm rev}.
    \label{eq:sym_edd_recursive_app}
\end{align}
Write
\begin{align}
    U_{\rm sym,EDD}^{(k-1)}(\tau)=e^{-iE_{k-1}(\tau)},
\end{align}
where $E_{k-1}(\tau)$ is the residual error action of the level-$(k-1)$ block. Assume that, for every $k$, the effective error action $E_{k-1}(\tau)$ belongs to a traceless operator space $\mathcal E$ that is averaged to zero by the base Eulerian cycle, that each leading pulse-width error generator also lies in $\mathcal E$, and that $\mathcal E$ is preserved under the recursive construction.

Under these assumptions, the first Magnus term of the outer cycle again vanishes by Eulerian averaging, while time symmetry again removes all even-order Magnus terms. Hence,
\begin{align}
    \|E_k(\tau)\|  = \mathcal O\big(l_{\rm EDD}^3(\|E_{k-1}(\tau)\|+\Lambda t_p)^3\big),
    \label{eq:sym_edd_error_recursion_app}
\end{align}
assuming Magnus convergence at each level. 

Under the additional smallness condition
\begin{align}
    l_{\rm EDD}^{3/2}\Lambda t_p \lesssim 1,
\end{align}
the recursion~\eqref{eq:sym_edd_error_recursion_app} can be iterated explicitly. Indeed, defining $x_k:=\|E_k(\tau)\|$ and $a:=\Lambda t_p$, Eq.~\eqref{eq:sym_edd_error_recursion_app} implies
\begin{align}
    x_k \le C l_{\rm EDD}^3 (x_{k-1}+a)^3 \le C' l_{\rm EDD}^3 (x_{k-1}^3 + a^3)
\end{align}
for suitable constants $C,C'>0$, where we used $(u+v)^3 \le 4(u^3+v^3)$. An
induction on $k$ then yields
\begin{align}
    \|U_{\rm sym,EDD}^{(k)}(\tau)-I\| = \mathcal O\Big(l_{\rm EDD}^{\frac{3^{k+1}-3}{2}}(\Lambda\tau)^{3^k} + l_{\rm EDD}^3(\Lambda t_p)^3 \Big).
    \label{eq:sym_edd_identity_bound_app}
\end{align}
Since $l_{\rm EDD}$ is at least the pulse count of the underlying first-order ideal sequence, it also generally grows exponentially with $n$, and so does the above bound. As in the CDD-based approach in the previous subsection, the $t_p$-independent error term achieves high-order suppression under concatenation. Unlike the CDD-based construction with DCGs, however, the $t_p$-dependent error term here is independent of $k$, although it is limited to cubic order in $t_p$.

\section{Error bounds for the multi-product formula implementations}
\label{app:mpf_bounds}
In this appendix, we combine Eq.~\eqref{eq:mpf_expectation_error} with the robust second-order blocks obtained from Constructions~1 and~2. For pairwise distinct integers $m_1,\dots,m_p\in\mathbb N$ and coefficients $b_1,\dots,b_p\in\mathbb R$ satisfying Eq.~\eqref{eq:mpf_condition}, let
\begin{align}
    U_j(T):=\bigl[\mathcal S_2(T/m_j)\bigr]^{m_j},
\end{align}
and let $\widetilde U_j(T)$ denote the corresponding implementation using finite-width pulses. Define
\begin{align}
    \widetilde\rho_j(T):=\widetilde U_j(T)\rho_0\widetilde U_j^\dagger(T),
\end{align}
and the hybrid MPF estimation error
\begin{align}
    \varepsilon_{\mathrm{MPF}}(O,T) := \left| \Tr\bigl(O\rho(T)\bigr) - \sum_{j=1}^p b_j \Tr\bigl(O\widetilde\rho_j(T)\bigr)\right|.
\end{align}
Then Eq.~\eqref{eq:mpf_expectation_error}, together with $\|\widetilde U_j\rho_0\widetilde U_j^\dagger-U_j\rho_0U_j^\dagger\|_1 \le 2\|\widetilde U_j-U_j\|$, implies that for any observable $O$ with $\|O\|\le 1$,
\begin{align}
    \varepsilon_{\mathrm{MPF}}(O,T) \le \mathcal O\bigl((\Gamma T)^{2p+1}\bigr) + 2\sum_{j=1}^p |b_j|\|\widetilde U_j(T)-U_j(T)\|.
    \label{eq:mpf_master_bound}
\end{align}

We first specialize the hybrid MPF error bound in Eq.~\eqref{eq:mpf_master_bound} to the robust second-order block obtained from Construction~1.

\begin{proposition} \label{prop:mpf_construction1}
For the second-order block obtained from Construction~1,
\begin{align}
    \varepsilon_{\mathrm{MPF}}(O,T) = \mathcal O\bigl((\Lambda T_c)^{2p+1}\bigr) + \mathcal O\left(l\varepsilon_{\mathrm{DCG}}^{(q)} \sum_{j=1}^p |b_j|m_j \right),
\end{align}
where $\varepsilon_{\mathrm{DCG}}^{(q)} = \mathcal O\bigl((\kappa^q\Lambda t_p)^{q+1}\bigr)$.
\end{proposition}

\begin{proof}
By Eq.~\eqref{exponent_algo1},
\begin{align}
    \Gamma T = \sum_{k=1}^l \|A_k\|T \le \Lambda \sum_{k=1}^l \tau_k = \Lambda T_c.
\end{align}
Moreover, each block $\mathcal S_2(T/m_j)$ contains $2l$ ideal pulses, so $U_j(T)=[\mathcal S_2(T/m_j)]^{m_j}$ contains $2lm_j$ such pulses. Replacing each by a $q$th-order DCG with error $\mathcal O(\varepsilon_{\mathrm{DCG}}^{(q)})$ and using a telescoping argument gives
\begin{align}
    \|\widetilde U_j(T)-U_j(T)\| = \mathcal O\bigl(lm_j\varepsilon_{\mathrm{DCG}}^{(q)}\bigr).
\end{align}
Substituting these bounds into Eq.~\eqref{eq:mpf_master_bound} proves the claim.
\end{proof}

We next specialize Eq.~\eqref{eq:mpf_master_bound} to the robust second-order block from Construction~2.

\begin{proposition}
\label{prop:mpf_construction2}
For the second-order block obtained from Construction~2, choose the stretching parameter $c=2m_{\max}$, where $m_{\max}:=\max_{1\le j\le p}m_j$. Then
\begin{align}
\begin{aligned}
    \varepsilon_{\mathrm{MPF}}(O,T) & = \mathcal O\left(
    \bigl(\Lambda(T_c+(2m_{\max}-1)lt_p)\bigr)^{2p+1}
    \right)
    \\
    & \qquad+ \mathcal O\left( lm_{\max}^2(\Lambda t_p)^2 \sum_{j=1}^p \frac{|b_j|}{m_j} \right).
\end{aligned}
\end{align}
\end{proposition}
\begin{proof}
By Eq.~\eqref{exponent_algo2},
\begin{align}
    \Gamma T &\le \sum_{k=0}^{l-1} \left(2m_{\max}\|\Phi_{k+1}^{(1)}\|+\|H_0\|\tau_{k+1}\right) \\
    &\le \Lambda\bigl(T_c+(2m_{\max}-1)lt_p\bigr),
\end{align}
where we used $\|\Phi_k^{(1)}\|\le \Lambda t_p$ and $T_c=\sum_{k=1}^l \tau_k+lt_p$. Hence the ideal MPF contribution in Eq.~\eqref{eq:mpf_master_bound} is
\begin{align}
    \mathcal O\left(\bigl(\Lambda(T_c+(2m_{\max}-1)lt_p)\bigr)^{2p+1}\right).
\end{align}
For the implementation error, write $\alpha=1/m_j$ and $\beta=\alpha m_{\max}=m_{\max}/m_j$. Then $\beta\ge 1$ for all $j$, so only the positive-$\beta$ relations in Eqs.~\eqref{pulse_1} and \eqref{pulse_2} are needed. Each pulse factor in one copy of $\mathcal S_2(T/m_j)$ is therefore implemented with error $\mathcal O\bigl((\Lambda\beta t_p)^2\bigr) = \mathcal O\left(\left(\Lambda\frac{m_{\max}t_p}{m_j}\right)^2\right)$. Since $\mathcal S_2(T/m_j)$ contains $2l$ pulse factors and $U_j(T)=[\mathcal S_2(T/m_j)]^{m_j}$ repeats this block $m_j$ times, a telescoping argument gives
\begin{align}
    \|\widetilde U_j(T)-U_j(T)\| &= \mathcal O\left( lm_j\left(\Lambda\frac{m_{\max}t_p}{m_j}\right)^2
    \right) \\
    &= \mathcal O\left( lm_{\max}^2(\Lambda t_p)^2\frac{1}{m_j} \right).
\end{align}
Substituting this into Eq.~\eqref{eq:mpf_master_bound} proves the claim.
\end{proof}

\section{Details on numerical experiments} \label{app:numerics}

\subsection{DCGs for the Ising-model numerics}\label{app:numerics_ising_dcg}

We now specialize the Eulerian DCG construction of Appendix~\ref{app:dcg} to the Ising-model numerics discussed in Sec.~\ref{sec:numerics_Ising}. 

\subsubsection{First-order DCG}
For the $n=3$ Ising-model numerics, the required ideal controls are
\begin{align}
    W\in\{X_1X_2,X_2X_3,X_1,X_2\}.
\end{align}
We implement each such $W$ using the first-order Eulerian DCG construction from Appendix~\ref{app:dcg}. With $\widetilde I_W=\widetilde W^{\rm rev}\widetilde W$ from Eq.~\eqref{I_W}, the pair $\{\widetilde I_W,\widetilde W(2t_p)\}$ is the associated first-order balance pair.

\begin{figure}[t!]
    \centering
    \includegraphics[width=8.6cm]{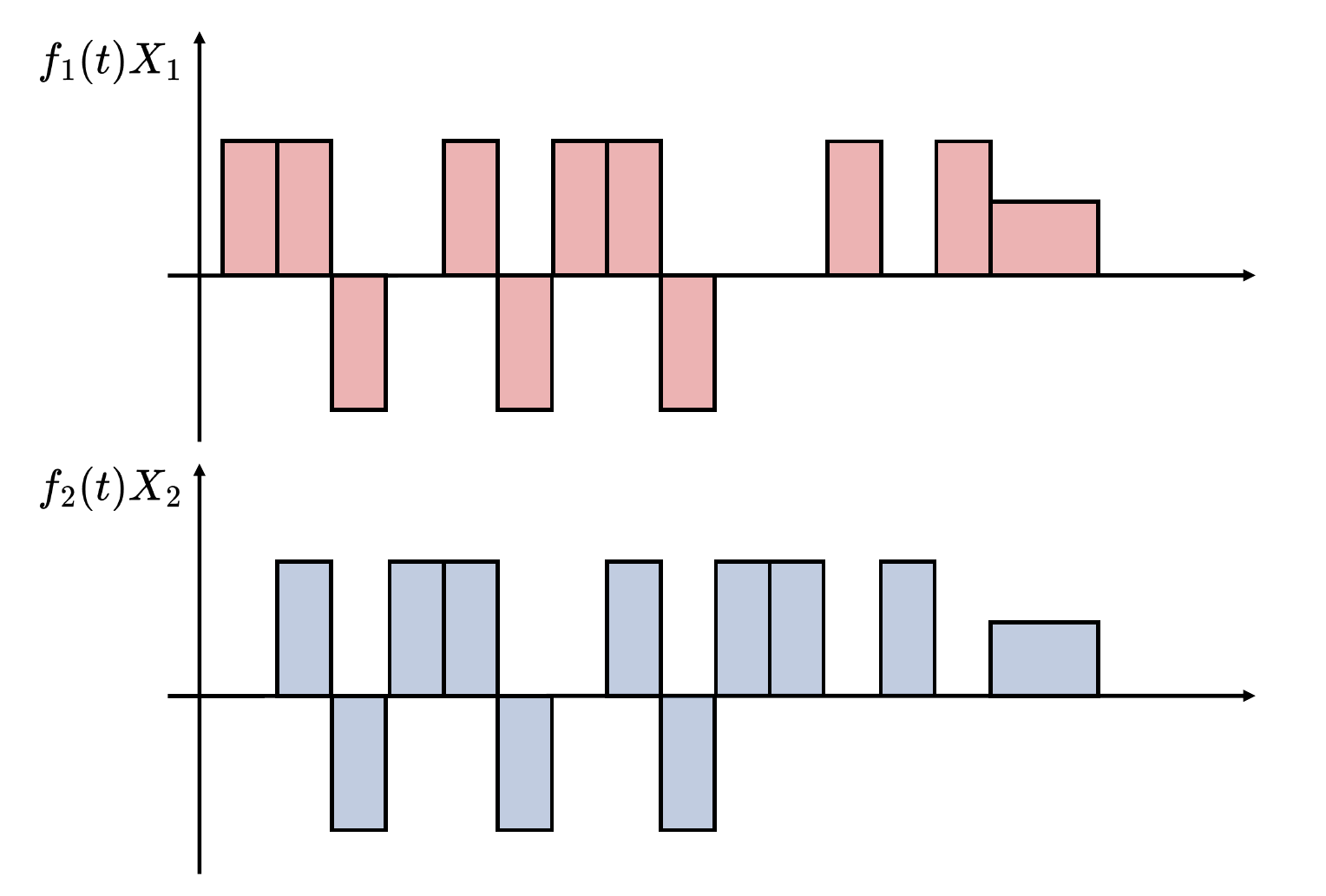}
    \caption{Rectangular-pulse implementation of the first-order DCG for $W=X_1X_2$. Narrow rectangles denote rectangular $\pi$ pulses of duration $t_p$ and amplitude $\pi/(2t_p)$. The upper and lower pulse pairs realize $\widetilde W$ or $\widetilde W^{\rm rev}$, and two consecutive such pairs realize $\widetilde I_W=\widetilde W^{\rm rev}\widetilde W$. The final wider pair represents the stretched pulse $\widetilde W(2t_p)$ with duration $2t_p$ and amplitude $\pi/(4t_p)$. The total control time is $16t_p$.}
    \label{fig:DCG_algo1_Ising}
\end{figure}

We choose
\begin{align}
    \Gamma_{\rm Ising}=\{X_1,X_2\},
    \quad
    \mathcal G_{\rm Ising}=\{I,X_1,X_2,X_1X_2\}.
\end{align}
Using the notation of Appendix~\ref{app:dcg}, let
\begin{align}
    \mathcal E_{\rm Ising} :=
    \mathrm{span}\bigl\{ \Phi_{X_1}^{(1)}, \Phi_{X_2}^{(1)}, \Phi_{X_1X_2}^{(1)}, \Phi_{X_2X_3}^{(1)} \bigr\}.
\end{align}
For $H_0 = J(Z_1Z_2+Z_1Z_3+Z_2Z_3)$, one finds
\begin{align}
\begin{aligned}
    \mathcal E_{\rm Ising}
    \subset
    \mathrm{span}\bigl\{
    &Z_1Z_2, Z_1Y_2, Y_1Z_2, Y_1Y_2, \\
    &Z_1Z_3, Y_1Z_3, Z_1Y_3,\\
    &Z_2Z_3, Y_2Z_3, Z_2Y_3, Y_2Y_3
    \bigr\}.
\end{aligned}
\end{align}
Each basis element contains a $Y$ or $Z$ on qubit 1 or 2, and hence changes sign under conjugation by $X_1$ or $X_2$. Therefore,
\begin{align}
    \frac{1}{|\mathcal G_{\rm Ising}|} \sum_{g\in\mathcal G_{\rm Ising}} g^\dagger E g =0, \qquad  \forall E\in\mathcal E_{\rm Ising},
\end{align}
so $\mathcal G_{\rm Ising}$ satisfies the decoupling condition in Eq.~\eqref{eq:dcg_decoupling_condition}.

Applying the augmented Cayley-graph construction with $\Gamma=\Gamma_{\rm Ising}$ and $\mathcal G=\mathcal G_{\rm Ising}$ therefore yields a first-order DCG for each $W$. Explicitly, 
\begin{align}\label{Ising_DCG_1st}
    \widetilde W^{[1]} = \widetilde W(2t_p) \widetilde X_1\widetilde X_2\widetilde X_1\widetilde X_2 \widetilde X_2\widetilde I_W\widetilde X_1\widetilde I_W\widetilde X_2\widetilde I_W\widetilde X_1.
\end{align}
Fig.~\ref{fig:DCG_algo1_Ising} illustrates the corresponding rectangular-pulse implementation for $W=X_1X_2$.

\subsubsection{Second-order DCG}
To obtain second-order DCGs, we enlarge the available local controls to include both $X$ and $Y$ rotations on the active qubits. For each target operation $W$, let $(a,b)$ denote the active qubit pair and let $c$ be the remaining spectator qubit. In the present numerics, we take $(a,b)=(1,2)$ for $W\in\{X_1X_2,X_1,X_2\}$ and $(a,b)=(2,3)$ for $W=X_2X_3$.

We choose
\begin{align}
    \Gamma_{ab}^{(2)}=\{X_a,Y_a,X_b,Y_b\},
\end{align}
and let $\mathcal G_{ab}^{(2)}$ be the corresponding two-qubit Pauli decoupling group on qubits $a$ and $b$. We define
\begin{align}
    \mathcal E_{ab}^{(2)} := \mathfrak{su}(4)_{ab}\otimes \mathrm{span}\{I_c,Z_c\}.
\end{align}
For $H_0= J(Z_1Z_2+Z_1Z_3+Z_2Z_3)$, the first-order error generators of all finite-width pulses appearing in the construction lie in $\mathcal E_{ab}^{(2)}$. Moreover, $[\mathcal E_{ab}^{(2)},\mathcal E_{ab}^{(2)}]\subseteq \mathcal E_{ab}^{(2)}$, so the residual error space is preserved under one level of concatenation.

Since $\mathcal G_{ab}^{(2)}$ is the full Pauli group on qubits $a$ and $b$, averaging over the group averages out every traceless operator on those qubits. Hence
\begin{align}
    \frac{1}{|\mathcal G_{ab}^{(2)}|} \sum_{g\in\mathcal G_{ab}^{(2)}} g^\dagger E g =  0, \qquad \forall E\in\mathcal E_{ab}^{(2)},
\end{align}
so $\mathcal G_{ab}^{(2)}$ satisfies the decoupling condition of Eq.~\eqref{eq:dcg_decoupling_condition} for both the first-order block and its first concatenation.

Let $\widetilde W^{[1]}$ denote the corresponding first-order Eulerian DCG using Eq.~\eqref{Ising_DCG_1st}, with total duration $\tau_1$. The associated balance pair from Eqs.~\eqref{cDCG_balance_pair_I}--\eqref{cDCG_balance_pair_W} is
\begin{align}
    \widetilde I_W^{[1]} & := \widetilde{W^\dagger}^{[1]}(\tau_1) \widetilde W^{[1]}(\sqrt{2}\tau_1), \\
    \widetilde W_*^{[1]} & := \widetilde W^{[1]}(\tau_1) \widetilde{W^\dagger}^{[1]}(\tau_1) \widetilde W^{[1]}(\tau_1).
\end{align}
Applying the augmented Cayley-graph construction again, now with generator blocks $\{\widetilde X_a^{[1]},\widetilde Y_a^{[1]},\widetilde X_b^{[1]},\widetilde Y_b^{[1]}\}$, self-loops labeled by $\widetilde I_W^{[1]}$, and an exit edge labeled by $\widetilde W_*^{[1]}$, therefore yields a second-order concatenated DCG $\widetilde W^{[2]}$ with $\|\widetilde W^{[2]}-W\|=\mathcal O((\Lambda t_p)^3)$. 

In the numerics, the corresponding pulse sequence is generated from an Eulerian path on the augmented Cayley graph of $\mathcal G_{ab}^{(2)}$ using Hierholzer's algorithm \cite{hierholzer1873ueber}.

\subsection{Details for the CR-model numerics}\label{app:numerics_cr}

We now specialize Construction~1 and the negative-time implementation of Sec.~\ref{sec:negative_time} to the cross-resonance numerics of Sec.~\ref{sec:numerics_CR}, focusing on the $n=4$ simulations in Fig.~\ref{fig:CR_result_main}. For convenience, write $Y_{1234}:=\sum_{i=1}^4 Y_i$, $Z_{1234}:=\sum_{i=1}^4 Z_i$, $Y_{24}:=Y_2+Y_4$, and $Z_{24}:=Z_2+Z_4$. Up to a global phase, the required rotations may be written as
\begin{align}
\begin{aligned}
    R_E^{\otimes 4}=e^{-i\pi Y_{1234}/4}e^{-i\pi Z_{1234}/4},\\
    H_{\rm even}=e^{-i\pi Y_{24}/4}e^{-i\pi Z_{24}/2}.
\end{aligned}
\end{align}

\subsubsection{First-order DCG}\label{app:numerics_cr_DCG}
We implement $R_E^{\otimes 4}$, $R_E^{\dagger\otimes 4}$, and $H_{\rm even}$ by first-order Eulerian DCGs. For the collective $Z$- and $Y$-rotations entering $R_E^{\otimes 4}$, define $W_{Z,1234}:=e^{-i\pi Z_{1234}/4}$, $W_{Y,1234}:=e^{-i\pi Y_{1234}/4}$, $Z_{12}:=Z_1Z_2$, and $Z_{23}:=Z_2Z_3$.

Applying the first-order DCG construction of Appendix~\ref{app:dcg} to $W_{Z,1234}$ with $\Gamma_Z=\{Z_{12},Z_{23}\}$ and $\mathcal G_Z=\{I,Z_{12},Z_{23},Z_{12}Z_{23}\}$, we introduce the error space
\begin{align}
    \mathcal E_{Z,1234} := \mathrm{span}\bigl\{ \Phi_{Z_{12}}^{(1)}, \Phi_{Z_{23}}^{(1)}, \Phi_{W_{Z,1234}}^{(1)} \bigr\}.
\end{align}
For the CR model $H_0=J(X_1Z_2+X_2Z_3+X_3Z_4)$, since conjugation by $Z_{12}$ or $Z_{23}$ mixes only $X_i$ and $Y_i$ on the acted-on sites, one finds
\begin{align}
\begin{aligned}
    \mathcal E_{Z,1234} \subset
    \mathrm{span}\bigl\{ X_1Z_2,  Y_1Z_2, X_2Z_3, Y_2Z_3, X_3Z_4,  Y_3Z_4 \bigr\}.
\end{aligned}
\end{align}
Each basis element changes sign under conjugation by at least one of $Z_{12}$ and $Z_{23}$, and hence
\begin{align}
    \frac{1}{|\mathcal G_Z|}
    \sum_{g\in\mathcal G_Z} g^\dagger E g =0,
    \qquad
    \forall E\in\mathcal E_{Z,1234}.
\end{align}
Thus $\mathcal G_Z$ satisfies the decoupling condition in Eq.~\eqref{eq:dcg_decoupling_condition} for $W_{Z,1234}$.

Similarly, for $W_{Y,1234}$ we take $\Gamma_Y=\{Y_2,Y_3\}$, $\mathcal G_Y=\{I,Y_2,Y_3,Y_2Y_3\}$, and $\mathcal E_{Y,1234}:=\mathrm{span}\{\Phi_{Y_2}^{(1)},\Phi_{Y_3}^{(1)},\Phi_{W_{Y,1234}}^{(1)}\}$. For the CR model, every $E\in\mathcal E_{Y,1234}$ contains an $X$ or $Z$ on qubit $2$ or $3$, and hence changes sign under conjugation by at least one of $Y_2$ and $Y_3$. Therefore $\mathcal G_Y$ also satisfies Eq.~\eqref{eq:dcg_decoupling_condition} for $W_{Y,1234}$.


\begin{figure}[t!]
    \centering
    \includegraphics[width=8.6cm]{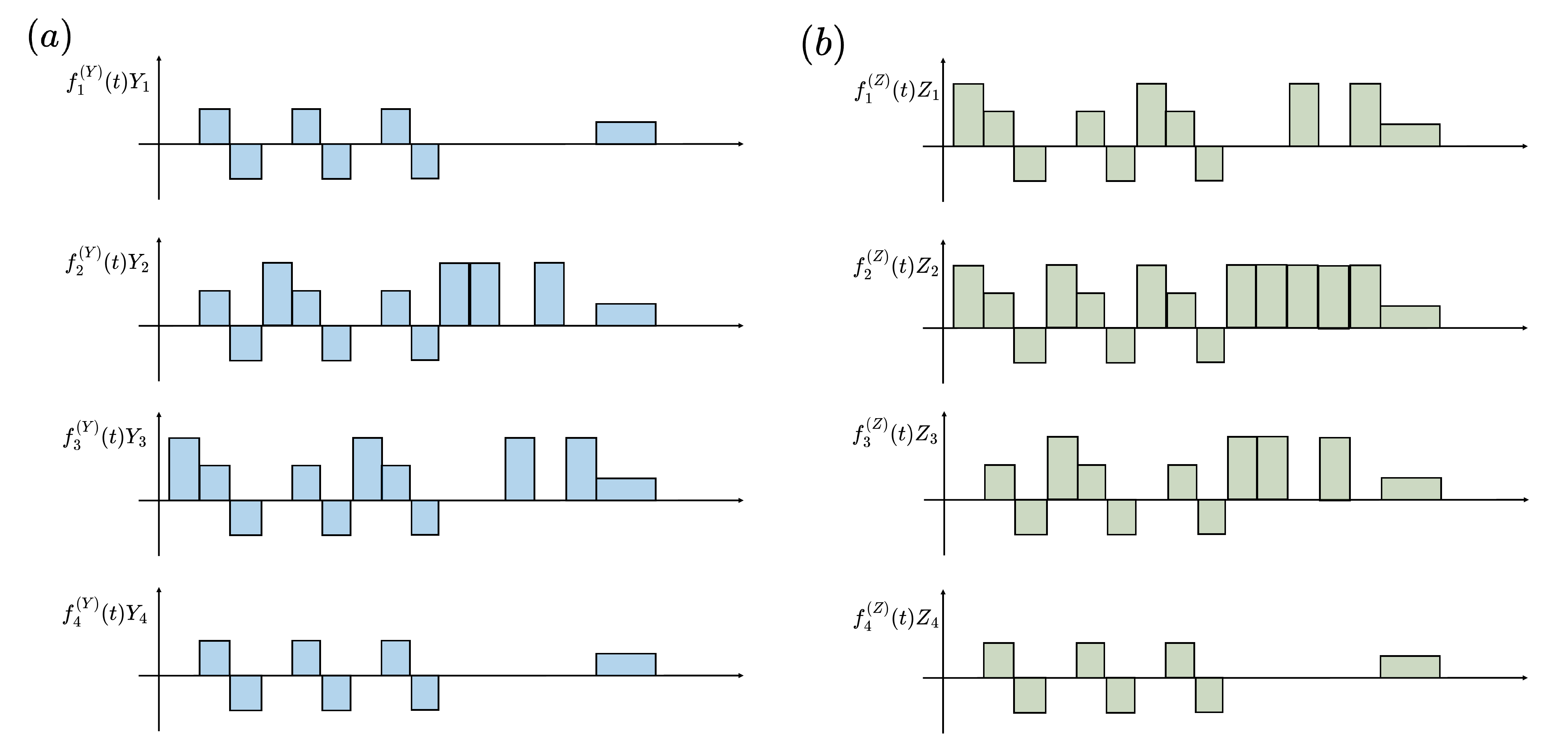}
    \caption{Rectangular-pulse implementations of $\widetilde W_{Y,1234}^{[1]}$ and $\widetilde W_{Z,1234}^{[1]}$ entering Eq.~\eqref{eq:cr_re_dcg}. Each sequence is obtained from Eq.~\eqref{eq:cr_first_order_dcg}. Composing the two blocks yields the robust implementation $\widetilde R_E^{[1]}$ of $R_E^{\otimes 4}$.}
    \label{fig:CR_DCG_pulse}
\end{figure}

Applying the augmented Cayley-graph construction to $(\Gamma_Z,\mathcal G_Z)$ and $(\Gamma_Y,\mathcal G_Y)$ yields first-order DCGs
\begin{align}\label{eq:cr_first_order_dcg}
    \widetilde W^{[1]} = \widetilde W(2t_p) \widetilde G_1\widetilde G_2\widetilde G_1\widetilde G_2 \widetilde G_2\widetilde I_W \widetilde G_1\widetilde I_W \widetilde G_2\widetilde I_W \widetilde G_1,
\end{align}
where $(G_1,G_2)=(Z_{12},Z_{23})$ for $W=W_{Z,1234}$ and
$(G_1,G_2)=(Y_2,Y_3)$ for $W=W_{Y,1234}$, and $\widetilde I_W=\widetilde W^{\rm rev}\widetilde W$ as in Eq.~\eqref{I_W}. We then implement
\begin{align}\label{eq:cr_re_dcg}
    \widetilde R_E^{[1]} := \widetilde W_{Y,1234}^{[1]}\widetilde W_{Z,1234}^{[1]},
    \quad
    \widetilde{R_E^\dagger}^{[1]} := \widetilde{W_{Z,1234}^\dagger}^{[1]}\widetilde{W_{Y,1234}^\dagger}^{[1]}.
\end{align}
Fig.~\ref{fig:CR_DCG_pulse} illustrates the corresponding rectangular-pulse implementations of the component DCGs $\widetilde W_{Y,1234}^{[1]}$ and $\widetilde W_{Z,1234}^{[1]}$ entering Eq.~\eqref{eq:cr_re_dcg}.

Analogously, we implement
\begin{align}
    \widetilde H_{\rm even}^{[1]}:=\widetilde W_{Y,24}^{[1]}\widetilde W_{Z,24}^{[1]},
\end{align}
where $H_{\rm even}=W_{Y,24}W_{Z,24}$, $W_{Y,24}:=e^{-i\pi Y_{24}/4}$, and $W_{Z,24}:=e^{-i\pi Z_{24}/2}$.

\subsubsection{Negative-time implementation}\label{app:numerics_cr_neg}

The fourth-order Trotter formula used in the CR numerics contains a negative-time free-evolution segment. We implement this segment using the $k=2$ instance of the symmetrized concatenated Eulerian DD construction from Appendix~\ref{subsec:appendix_neg_time_eulerian}.

At the first level, we take $\Gamma^{(1)}=\{Y_2,Y_3\}$. An Eulerian cycle on the corresponding Cayley graph is $\{Y_3,Y_2,Y_3,Y_2,Y_2,Y_3,Y_2,Y_3\}$, which defines the first-level identity block $U_{\rm sym,EDD}^{(1)}(\tau)$ through Eq.~\eqref{eq:sym_edd_base_app}.

\begin{figure}[t!]
    \centering
    \includegraphics[width=8.6cm]{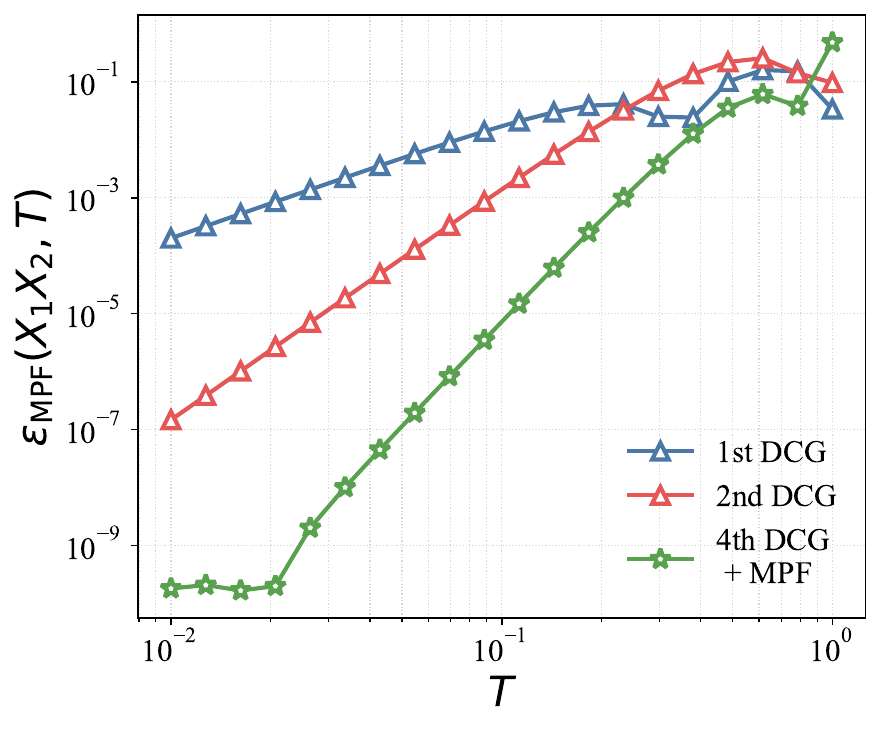}
    \caption{Hybrid multi-product-formula results. For $n=4$ and $t_p=10^{-5}$, we plot the observable-estimation error for $O=X_1X_2$ and $\rho_0=\ket{0101}\bra{0101}$. The MPF combination \eqref{MPF_M2} built only from robust second-order blocks, exhibits the expected fifth-order scaling in $T$.}
    \label{fig:Algo1_MPF}
\end{figure}

At the second level, the Eulerian cycle must decouple the residual error action of the first-level block. To ensure this, we enlarge the generating set to $\Gamma^{(2)}=\{X_i,Y_i\}_{i=1}^4$, which generates the four-qubit Pauli decoupling group. Using an Eulerian cycle on the corresponding Cayley graph, constructed with Hierholzer's algorithm, and substituting it into Eq.~\eqref{eq:sym_edd_recursive_app}, we obtain the second-level identity block $U_{\rm sym,EDD}^{(2)}(\tau)$. Writing
\begin{align}
    U_{\rm sym,EDD}^{(2)}(\tau)=\prod_{j=1}^{l_{\rm neg}} Q_j e^{-iH_0\tau},
\end{align}
where $Q_j$ denotes the $j$th pulse in this level-2 symmetrized EDD cycle, the negative-time propagator used in the numerics is defined, following Eq.~\eqref{eq:neg_dd_block}, by
\begin{align}\label{eq:cr_negative_time_block}
    U_{\rm neg}^{(2)}(\tau):=\left(\prod_{j=2}^{l_{\rm neg}} Q_j e^{-iH_0\tau}\right)Q_1, \qquad \tau>0.
\end{align}
Defining $u_2=(4-4^{1/3})^{-1}$ and $\tau_-:=(4u_2-1)T>0$, the negative free-evolution segment
\begin{align}
    e^{-iH_0(1-4u_2)T/2}=e^{+iH_0\tau_-/2}
\end{align}
appearing in the fourth-order block is implemented by $U_{\rm neg}^{(2)}(\tau_-/2)$.

\begin{figure}[t!]
    \centering
    \includegraphics[width=0.95\linewidth]{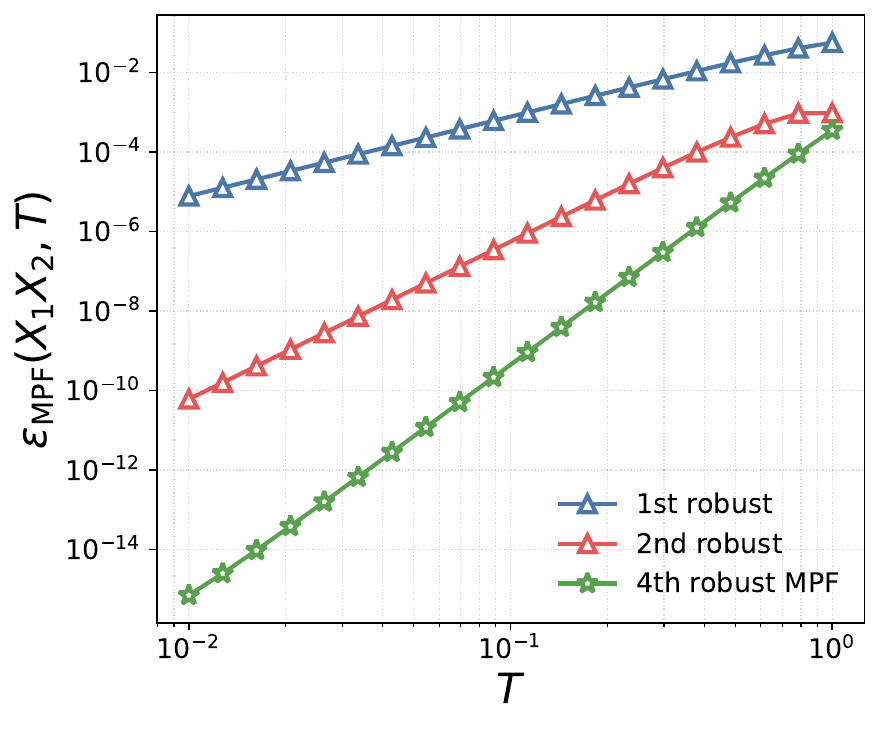}
    \caption{Hybrid MPF error for the anisotropic Heisenberg example. We use the same $p=2$ MPF combination as in Eq.~\eqref{MPF_M2}, with the symmetrized second-order sequence $\mathcal S_2(T)$ for the present model. The plot shows the expected fifth-order scaling in $T$ while using only positive-time robust second-order sequences.}
    \label{fig:heisenberg_mpf}
\end{figure}

\subsubsection{Hybrid multi-product-formula numerics}\label{app:numerics_cr_mpf}

We briefly benchmark the hybrid multi-product formula of Sec.~\ref{subsec:mpf}. We use the $p=2$ formula \cite{childs2012hamiltonian}
\begin{align}\label{MPF_M2}
    M^{(2)}(T) =  \frac{4}{3}\big[\mathcal S_2(T/2)\big]^2 -\frac{1}{3}\mathcal S_2(T).
\end{align}
For the numerics, we take the initial state $\rho_0=\ket{0101}\bra{0101}$ and the observable $O=X_1X_2$. With the notation of Appendix~\ref{app:mpf_bounds}, we evaluate $\varepsilon_\textrm{MPF}(O,T):= |\Tr(O\rho(T)) - \frac{4}{3}\Tr\big(O\widetilde{\rho}_2(T)\big)+\frac{1}{3}\Tr(O\widetilde{\rho}_1(T))|$ where $\widetilde{\rho}_1(T) = \widetilde S_2(T) \rho_0 \widetilde{S}_2^\dagger(T)$ and $\widetilde \rho_2(T) = \big[\widetilde{\mathcal S}_2(T/2)\big]^2 \rho_0 \big[\widetilde{\mathcal S}_2^\dagger(T/2)\big]^2$. Fig.~\ref{fig:Algo1_MPF} shows that the MPF combination, while using only positive-time robust second-order sequences, achieves the expected fifth-order scaling in $T$. 

\subsection{Hybrid MPF for the anisotropic Heisenberg model}
\label{app:numerics_heisenberg_mpf}

We benchmark the hybrid MPF for the anisotropic Heisenberg example. We use the same $p=2$ formula as in Eq.~\eqref{MPF_M2}, together with the same initial state, observable, and error metric $\varepsilon_{\mathrm{MPF}}(O,T)$ as in Appendix~\ref{app:numerics_cr_mpf}. Here each term is built from the symmetrized second-order sequence for the present example evaluated at time $T/m_j$, with the corresponding pulse stretching prescribed by the hybrid-MPF construction. For $p=2$, this means using $m_j\in\{1,2\}$, with no stretching for the $[\mathcal S_2(T/2)]^2$ branch and a stretch factor of $2$ for the $\mathcal S_2(T)$ branch.

\bibliography{ref_upd}

\end{document}